\documentclass[journal,twoside,web]{ieeecolor}

\usepackage{jsen}
\usepackage{cite}
\usepackage{amsmath,amssymb,amsfonts}
\usepackage{graphicx}
\usepackage{textcomp}
\usepackage{wrapfig}
\RequirePackage[latin1,utf8]{inputenc}
\usepackage[T1]{fontenc}
\usepackage{graphicx}
\usepackage{cite}
\usepackage{picinpar}
\usepackage{amsmath}
\usepackage{amssymb,bm}
\usepackage{float}
\usepackage{stfloats}
\usepackage{url}
\usepackage{algorithmicx}
\usepackage{algpseudocode}
\usepackage[latin1]{inputenc}
\usepackage{colortbl}
\usepackage{soul}
\usepackage{multirow}
\usepackage{pifont}
\usepackage{color}
\usepackage{alltt}
\usepackage{enumerate}
\usepackage{hyperref} 
\usepackage{epstopdf}
\usepackage{pbox}
\usepackage[font=footnotesize]{caption}
\usepackage[font=footnotesize]{subcaption}
\usepackage{comment}

\newtheorem{lemma}{Lemma}
\usepackage{lipsum}
\usepackage{tikz}
\usetikzlibrary{calc}     
\usetikzlibrary{angles,quotes,positioning,fit,arrows.meta}
\usepackage{algorithm}
\usepackage{algpseudocode}
\usepackage{siunitx}
\usepackage{graphicx,subcaption,booktabs,algorithm,algpseudocode}

\begin{document}
\title{Wavelet-Guided Water-Level Estimation for ISAC}
\author{
{Ayoob~Salari, Kai~Wu,~\IEEEmembership{Senior~Member,~IEEE}, Khawaja~Fahad~Masood,~\IEEEmembership{Member,~IEEE},\\
Y. Jay~Guo,~\IEEEmembership{Fellow,~IEEE}, and J. Andrew~Zhang,~\IEEEmembership{Senior~Member,~IEEE} }
\thanks{Ayoob Salari, Kai Wu, Fahad Masood, Y. Jay Guo and J. Andrew Zhang are with Global Big Data Technologies Centre, the University of Technology Sydney, Ultimo, NSW 2007, Australia (e-mail: \{ayoob.salari, kai.wu, khawajafahad.masood, jay.guo, andrew.zhang\}@uts.edu.au).}
}
\IEEEtitleabstractindextext{%
\begin{minipage}{\textwidth}%
\begin{abstract}
Real-time water-level monitoring across many locations is vital for flood response, infrastructure management, and environmental forecasting. Yet many sensing methods rely on fixed instruments—acoustic, radar, camera, or pressure probes—that are costly to install and maintain and are vulnerable during extreme events. We propose a passive, low-cost water-level tracking scheme that uses only LTE downlink power metrics reported by commodity receivers. The method extracts per-antenna RSRP, RSSI, and RSRQ, applies a continuous wavelet transform (CWT) to the RSRP to isolate the semidiurnal tide component, and forms a summed-coefficient signature that simultaneously marks high/low tide (tide-turn times) and tracks the tide-rate (flow speed) over time. These wavelet features guide a lightweight neural network that learns water-level changes over time from a short training segment.
Beyond a single serving base station, we also show a multi-base-station cooperative mode: independent CWTs are computed per carrier and fused by a robust median to produce one tide-band feature that improves stability and resilience to local disturbances. Experiments over a 420 m river path under line-of-sight conditions achieve root-mean-square and mean-absolute errors of 0.8 cm and 0.5 cm, respectively. Under a non-line-of-sight setting with vegetation and vessel traffic, the same model transfers successfully after brief fine-tuning, reaching 1.7 cm RMSE and 0.8 cm MAE. Unlike CSI-based methods, the approach needs no array calibration and runs on standard hardware, making wide deployment practical. When signals from multiple base stations are available, fusion further improves robustness.

\end{abstract}

\begin{IEEEkeywords}
ISAC, LTE received-power, Tide, RSRP, RSSI, RSRQ, Water Level Sensing, Wavelet Transformation.
\end{IEEEkeywords}
\end{minipage}
}

\maketitle

\section{Introduction}
\label{sec:introduction}
%

\IEEEPARstart{F}{loods} are the most common and most expensive natural disasters on the planet.  Between 2000 and 2019 they claimed more than 125 000 lives and caused about US\$650 billion in direct damage worldwide \cite{GAR19}.  Recent floods in New South Wales (NSW), Australia, show the local danger: the 2022–23 La Niña season broke a dozen river-level records, triggered the longest continuous state emergency service deployment on record, and left the insurance sector with losses of roughly AU\$5 billion \cite{ICA23}. Up-to-the-minute water-level information is therefore essential for issuing public warnings, closing roads and rail lines, planning dam releases, and modelling insurance exposure \cite{WMO21}.  Unfortunately, the fixed sensor network is still too sparse in many regions, and sensor themselves are costly to install, hard to maintain, and often swept away by the very floods they monitor.  Expanding coverage with low-cost, robust sensing has become a headline target of the United Nations “Early Warnings for All” initiative \cite{UN22}.


Conventional water-level monitoring still relies mainly on fixed sensors placed in the river.  The simplest is the staff gauge: a painted ruler bolted to a bridge pier that operators read manually or with a camera.  A step up is the vented pressure sensor, a small probe lying on the riverbed whose internal vent tube removes barometric effects so that the recorded pressure translates directly to depth.  Both devices need regular cleaning, checks that the height reference (the datum) has not shifted, and occasional flushing of the pipe that shelters the probe and damps out waves (often called a stilling well)~\cite{Rantz1982,WMO2010}.
To cut this maintenance, many sites now fit non-contact acoustic rangers: an ultrasonic head under the bridge sends a short tone and times the echo from the water surface~\cite{Muste2016}. Low-power radar units work in the same way but are immune to air temperature and condensation, and are becoming common on bridges and small overflow structures~\cite{Birkholzer2020}.  Satellite altimeters such as Jason-3 or Sentinel-3 give global coverage, but only every 9–27 days and with 10–20 cm accuracy once river-slope corrections are applied~\cite{Shao2019}. 
The newest alternative is camera sensing: a bank-mounted RGB camera tracks a painted marker or shoreline texture and can reach 2–3 cm when light is good, though accuracy falls in rain or fog~\cite{Hou2019}.  


Despite steady progress, practical obstacles still limit how far the above instruments can be rolled out.  First is cost: a pressure-probe or radar node, its solar supply, and a telemetry modem together exceed US\$2000, and dependable flood mapping needs many of them along every reach \cite{Muste2016}.  Once installed, storm maintenance becomes the dominant expense—debris and sediment clog stilling wells, ultrasonics mist over, and floating logs can shear an entire gauge from its mount \cite{Rantz1982}.  Even where budgets allow, coverage gaps remain: the United States operates about 8000 river gauges for 5 million kilometres of waterways, leaving most urban catchments unmonitored \cite{USGS2023}.  At the other extreme, satellite altimeters see every major river but only every 9–27 days; their 10–20 cm uncertainty is acceptable for climatology yet too coarse for flash-flood operations \cite{Shao2019}.  Camera systems are affordable, but heavy rain, fog, or power cuts can disable them when water levels rise \cite{Hou2019}.  
In short, current methods struggle to be simultaneously dense, robust, and cheap. These limitations motivate the present study: repurposing the already-present 4G/5G radio  infrastructure—towers, handsets, and routine control signalling—as a complementary, low-maintenance water-level sensor network.


Mobile broadband networks present an attractive sensing platform that already covers most populated terrain.  The three
Australian operators, for example, run more than 27000 4G/5G cells— on average one site every 4 km in metropolitan New South Wales, Australia —and similar densities are reported across Europe and North America \cite{GSMAInfra2024}.  Base station panels are typically mounted 15–40 m above ground on rooftops or monopoles, safely above even major flood crests, and every site carries standby power specified to survive at least \SI{8}{h} of outage \cite{ACMAResilience2022}.  
Importantly, coverage is overlapping: a suburban handset sees three to six macro-cells from different operators and sectors, so the loss of any one tower leaves multiple independent links that can used.  Because environmental information is embedded in routine downlink control signals — reference-signal received power (RSRP), signal strength (RSSI), and quality (RSRQ)— no extra spectrum, user handset modification, or licensed transmitter is required; only passive measurements at the receiver end \cite{PoleseComSensReview2022}.  These traits—dense footprint, flood-line elevation, built-in redundancy, and zero spectrum cost—make cellular infrastructure a compelling candidate for augmenting conventional water monitoring.


Several RF-based studies have shown that environmental conditions can be inferred from changes in wireless signal propagation. Rainfall estimation has been demonstrated using commercial microwave backhaul links by correlating received signal strength with attenuation \cite{Leijnse07Rain,Overeem2013RainLink,Kjeldsen2021Rain5G}. Other work has exploited channel-state information (CSI) from Wi-Fi and LTE testbeds to estimate soil moisture \cite{Chen2017SoilCSI,Li2020LoRaMoist}, snow depth \cite{Larson2009SnowGNSS}, and small-scale water levels in controlled tanks \cite{Guo2021CSIWater}. Passive backscatter tags have achieved millimetre-precision in lab-scale flood gauging \cite{Chu2018BackscatterFlood}, and low-power LoRa nodes have tracked tidal fronts in salt marshes \cite{Li2023LoRaTide}.
Among these, the most comprehensive CSI-based framework to date is PMNs-WaterSense \cite{wang2025water}, which extracts Doppler–delay–angle features from LTE and mmWave CSI using multi-domain filtering and Kalman-based phase unwrapping. While it reports sub-centimetre precision in controlled lab tests, its outdoor trial yields a 4.8 cm mean error and requires 30-minute CSI windows. More critically, the approach depends on raw CSI capture via SDR platforms and calibrated multi-antenna arrays—tools not available on commodity LTE devices.
Despite their promise, these methods face three key barriers to wide-scale flood sensing: (i) reliance on high-rate CSI or custom hardware, (ii) sensitivity to calibration and testbed setup, and (iii) limited multi-site, long-term validation. In contrast, our work uses only LTE power indicators (RSRP, RSSI, RSRQ) available on unmodified devices, introduces a wavelet-guided feature that isolates tidal dynamics, and demonstrates generalisation across two geometrically distinct river sites captured months apart.


In this work, we propose a wavelet-guided water level estimation framework that leverages downlink LTE signal power—specifically RSRP, RSSI, and RSRQ—as a passive hydrological sensing modality. Unlike CSI-based methods that rely on SDR hardware, wideband measurements, and antenna calibration, our system uses only standards-compliant LTE power metrics, accessible via commercial modems and handsets. We begin with a physical two-ray model that analytically links tide height to received power, and then introduce a wavelet-based feature extractor that identifies tide-driven power variations with \SI{5}{min} latency. These features, feed into a lightweight neural regressor that continuously estimates water level with centimetre-scale accuracy—without requiring transmitter-side control, infrastructure changes, or offline retraining. In addition, we extend the wavelet stage to fuse measurements from multiple base stations, which increases robustness while preserving the same online detector.

Our main contributions are:
\begin{enumerate}
  \item 
  We present the first field-scale study to passively monitor tidal water levels using RSRP, RSSI, and RSRQ from a live LTE macro-cell. The receiver setup includes multi-antenna SDRs facing a 35~m high base station across a 420~m river span. Unlike prior lab-scale studies using CSI, our work demonstrates that standard LTE control-channel power measurements are sufficiently sensitive to detect tidal effects in realistic outdoor environments.
  
  \item 
  We derive and validate a continuous wavelet transform (CWT) feature, that aggregates coefficients within the semidiurnal tide band and analytically shows it tracks the absolute value of the tide-change rate. This feature marks high/low tide (tide-turn times) and peak-flow instants in real time using a forward-buffered minimum/maximum detector. The wavelet method is mathematically linked to the two-path power model and is robust to boat wakes, clutter, and short data gaps.

  \item 
  Using a single-hidden-layer neural network trained on engineered wavelet and geometric features, we achieve \SI{0.8}{cm} RMSE in LOS tests and \SI{1.7}{cm} RMSE under obstructed, vegetation-covered NLOS settings. Crucially, we demonstrate cross-date generalisation: a model trained on one site and date transfers to a different day and geometry with no architectural change and only brief fine-tuning, highlighting the geometry-invariant structure of the features.

  \item We extend the wavelet stage to a multi–base-station setting and demonstrate median-fusion of per-carrier CWT magnitude sums from two operators. The fused tide-band feature follows the tide-rate while preserving clear maxima at peak flow and minima at tide-turn times.

  \item 
  We conduct an additional deployment at a separate site with co-located sonar ground truth and stable mains power, enabling multi-hour to multi-day captures. Continuous wavelet analysis of the water-level record directly recovers the semidiurnal tidal period, and the RSRP scalogram exhibits a stable low–millihertz energy band consistent with tide-driven propagation. Building on this long-duration dataset, we develop a compact feed-forward model that accurately estimates water level from LTE power features; on a two-day held-out test interval, it closely tracks the tidal envelope.

  \item 
  We release a complete end-to-end software stack—from LTE decoding and power extraction, through wavelet processing and feature engineering, to lightweight neural inference. The pipeline enables other researchers or agencies to replicate, adapt, or deploy our system using affordable SDR hardware and commodity LTE downlink signals, without requiring access to CSI or licensed transmitters.
\end{enumerate}

The remainder of this paper is structured as follows. Section~\ref{sec:sysmodel} introduces the physical geometry and derives a two-ray model linking received power to tide height. Section~\ref{sec:wavelet} develops the wavelet representation that isolates tide-driven content, defines the summed-coefficient feature linked to tide-rate, and extends the method to fuse measurements from multiple base stations. Section~\ref{sec:exp} presents the experimental study, covering field deployments (6~Sep, 2~Dec, and a long-duration capture), preprocessing and wavelet diagnostics, feature engineering, the neural estimator, and quantitative results for same-day, cross-date, and multi–base-station settings. Section~\ref{sec:concl} concludes the paper.

%
%
%
%
\section{System Model}\label{sec:sysmodel}

We model how the tide height \(h(t)\) (measured from mean sea level) changes the received power in a line-of-sight radio link.  Geometry and hardware numbers are left for Section~\ref{sec:exp}.

\begin{figure*}[t]
  \centering
  \begin{subfigure}[b]{0.5\textwidth}
    \centering
    \includegraphics[width=\linewidth]{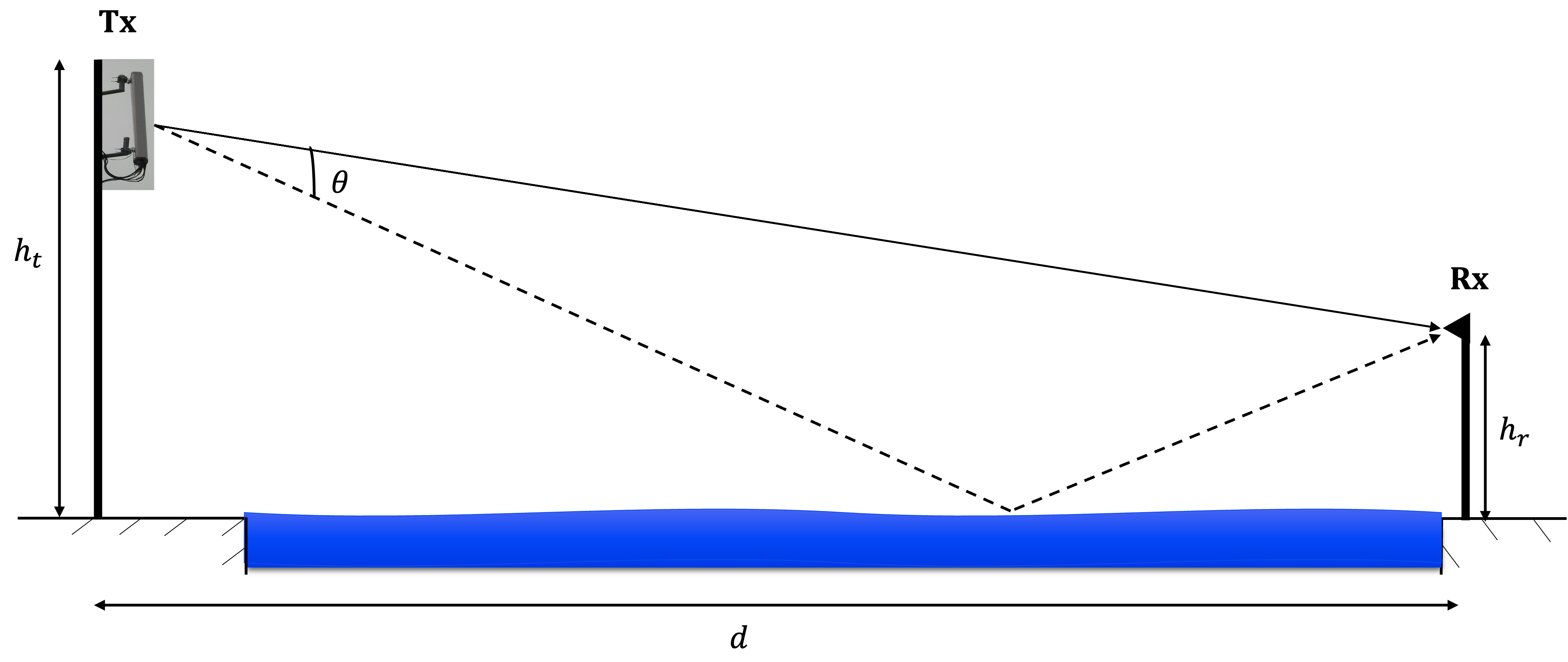}
    \caption{Two-ray geometry with direct and single-bounce paths.}
    \label{fig:tworay}
  \end{subfigure}
  \begin{subfigure}[b]{0.45\textwidth}
    \centering
    \includegraphics[width=\linewidth]{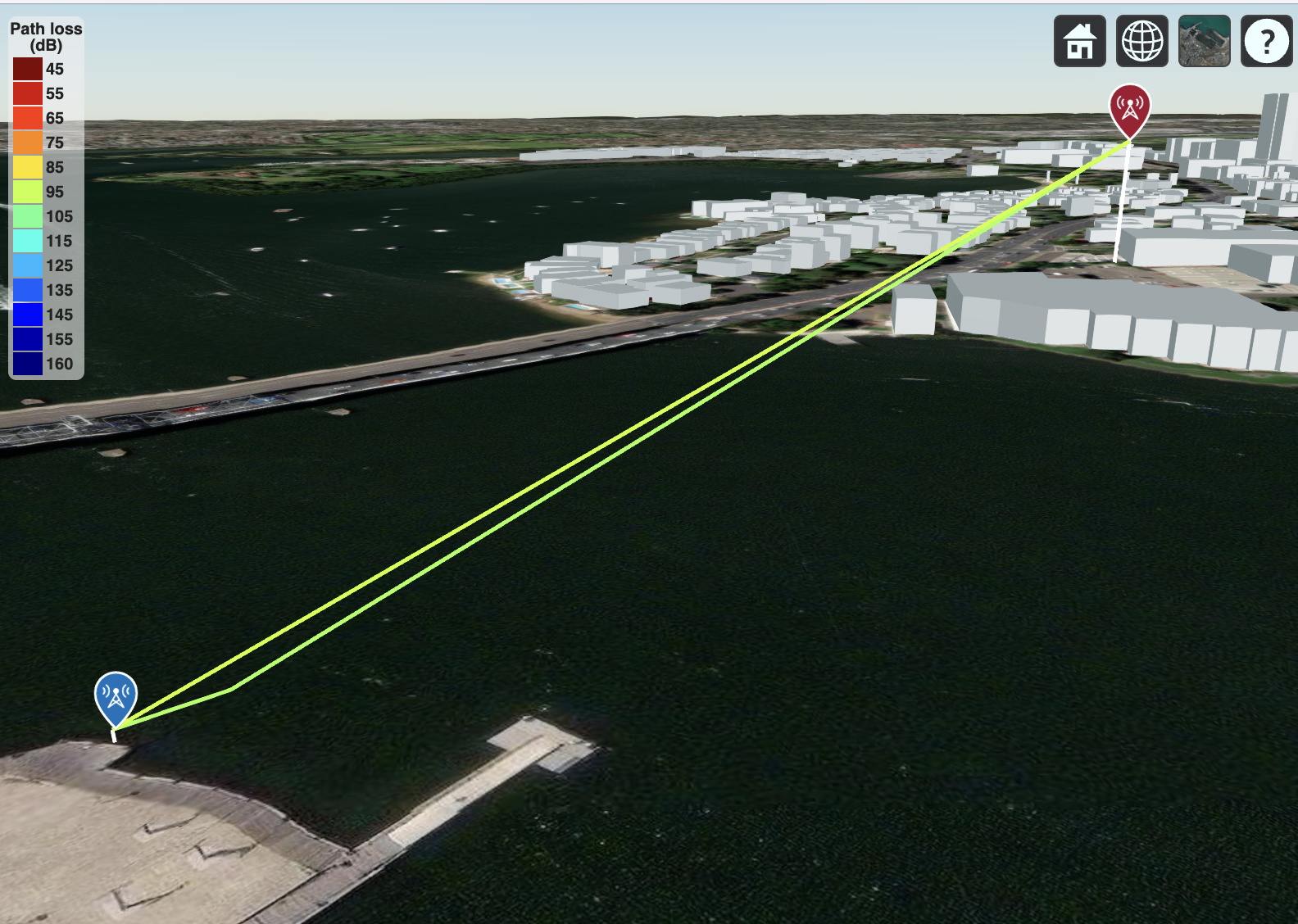}
    \caption{Matlab Raytracing modelling of Experimental setup.}
    \label{fig:raytracing}
  \end{subfigure}

  \caption{Experiment set-up: (a) geometry, (b) Raytracing model}
  \label{fig:setup}
\end{figure*}

\subsection{Two-Ray Channel}

A transmitter at height \(h_t\) and a receiver at height \(h_r\) face each other across a water surface separated by horizontal distance \(d\) (Fig.~\ref{fig:tworay}).  
The surface is smooth compared with the first Fresnel zone, so one specular\footnote{Mirror-like: angle of incidence equals angle of reflection.} echo is enough to capture the dominant multipath~\cite{Balanis2021}.\\
With one direct and one reflected path the complex baseband signal on receive antenna \(k\) is
\begin{equation}
r_k(t)=\frac{A_k}{d}\,e^{-j 2\pi d/\lambda}
\Bigl[1+\rho(\theta)
      \,e^{-\tfrac{2\pi}{\lambda}\,\Delta d(h(t))}\Bigr]+n_k(t),
\label{eq:rx_cplx}
\end{equation}
where \(\lambda\) is the wavelength and
\begin{equation}
\Delta d(h)=
\sqrt{d^{2}\!+\!\bigl(h_t+h_r-2h\bigr)^{\!2}}
-\sqrt{d^{2}\!+\!(h_t-h_r)^{\!2}}.
\label{eq:delta_exact}
\end{equation}
Because \(d\gg h_t,h_r\) the square roots expand in first order to
\begin{equation}
\Delta d(h)\;\simeq\;\frac{2h\,(h_t-h_r)}{d}.
\label{eq:delta_linear}
\end{equation}

\subsection{Received-Power Variation}

Taking the squared magnitude of \eqref{eq:rx_cplx} gives
\begin{equation}
P_{r,k}(h)=
P_{0,k}\Bigl[1+|\rho|^{2}
      +2|\rho|\cos\!\bigl(\tfrac{2\pi}{\lambda}\,\Delta d(h)+\phi_\rho\bigr)\Bigr],
\label{eq:prh}
\end{equation}
where \(\phi_\rho=\arg\rho(\theta)\).  
Modern LTE and 5 G panels use \(\pm45^{\circ}\) cross-polar ports, so the field has equal vertical and horizontal parts.  Near the horizon (\(\theta\!<\!10^{\circ}\)) the vertical part reflects with phase \(\pi\) and
the horizontal part with phase \(0\)~\cite{3GPP38901}.  The sum is still almost unit amplitude with phase \(\approx\pi\); hence we take \(|\rho|\simeq1\) and \(\phi_\rho\simeq\pi\). \\
A full \(2\pi\) cycle of the cosine happens when the tide shifts by
\begin{equation}
\Delta h_{\text{cycle}}=\frac{\lambda\,d}{2\,(h_t-h_r)}.
\label{eq:cycle}
\end{equation}
Because \(\lambda\) is smaller at higher carrier frequency, which in turn improves the resolution
with which gradual water-level changes can be tracked.
Equations \eqref{eq:prh}–\eqref{eq:cycle} show directly that the received
power parameters (RSRP, RSSI, RSRQ) vary in step with the instantaneous
tide height \(h(t)\).


\subsection{Simulation–Experiment Cross-Check}

The two-ray model in~\eqref{eq:rx_cplx} assumes  
(i) a flat water surface that acts as a single mirror,
(ii) only one dominant reflected path, and  
(iii) no time variation other than the slow change in tide height.  
Before relying on these assumptions, we verify them against a river-bank measurements.

\paragraph{Simulator.}  
Using the link geometry \((h_t,h_r,d)\), the wavelength \(\lambda\), and a specified tide profile \(h(t)\), we developed a simulator that computes the complex voltages of the direct and single-bounce paths for each water-height sample, superposes them, and converts the power to RSRP.
No fitting, calibration, or learning is involved; the output depends solely on the physics captured by the two-ray equations.

\paragraph{Comparison with field data.}  
Figure~\ref{fig:sim_vs_meas} overlays the simulated envelope (red) on LTE RSRP measured next to the river (blue; instrumentation details are given in Section~\ref{sec:exp}). Across the full range of tide heights, the simulated curve aligns with the measured values.
The small point-to-point scatter is expected—secondary reflections from banks and moving boats are not included in the first-order model.

\begin{figure}[t]
  \centering
  \includegraphics[width=\linewidth]{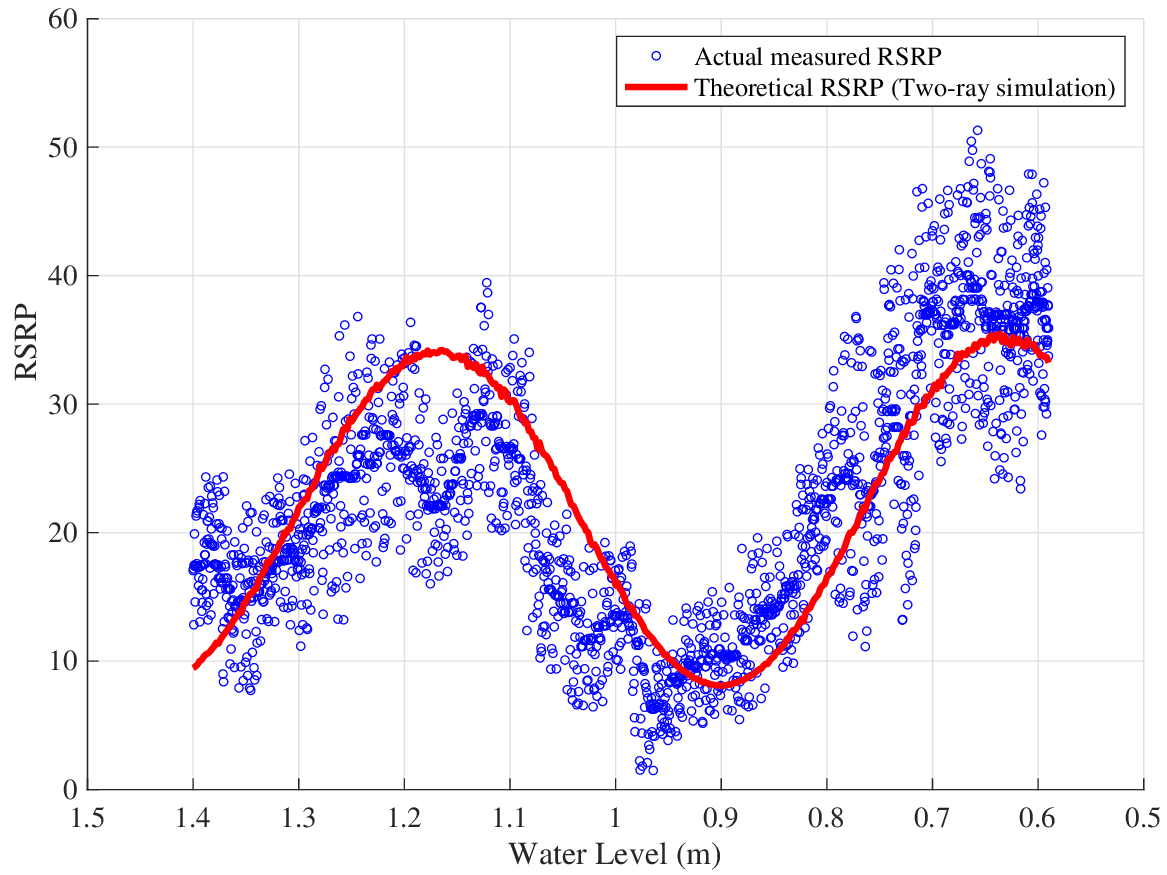}
  \caption{Two-ray simulation (red) versus measured LTE RSRP (blue) as a function of tide height. The match confirms that one reflected path captures the dominant dependence on water level.}
  \label{fig:sim_vs_meas}
\end{figure}

The agreement confirms that variations in received power track the tide height in a predictable way.  The next section introduces a wavelet-based framework that isolates the tide-driven component and pinpoints the moments when the tide reverses.

%
%
%

\section{Wavelet-Based Tide-Feature Extraction}\label{sec:wavelet}

Received-power traces exhibit two very different time scales: a slow, almost sinusoidal drift imposed by the tide and short-lived excursions caused by boats, fading, or transmitter drift.  A representation that resolves a signal both in time and in frequency is therefore required.  The CWT provides that capability and is widely used for geophysical and biomedical records with comparable behaviour \cite{Mallat1999}.

We use the CWT because it handles our tide-affected signal recordings better than any other method. A short-time Fourier transform must choose one fixed window, forcing a compromise between time and frequency detail. The CWT sidesteps that compromise: its window stretches at large scales to follow the slow tide envelope and contracts at small scales to capture the rapid water-level changes. This flexible time–frequency view reveals exactly when the tide-driven part of the signal strengthens or weakens. Energy at tide scales stands out clearly, while rapid disturbances from boats, wind ripples, or RF noise shift to higher pseudo-frequencies and are easy to filter out. Because each wavelet analyses only a short segment of data, brief gaps from SDR retuning affect only that local region instead of corrupting the whole record. Finally, CWT implementations that keep every sample (without down-sampling) run in time proportional to the data length, enabling real-time processing on modest embedded hardware.

For a real-valued trace \(x(t)\) the CWT is
\begin{equation}
W(a,b)=\int_{-\infty}^{\infty}
        x(t)\,\frac{1}{\sqrt{|a|}}\,
        \psi^{*}\!\Bigl(\tfrac{t-b}{a}\Bigr)\,\mathrm{d}t ,
\label{eq:cwt_def}
\end{equation}
where  
\begin{itemize}
\item \(a>0\) is the scale: large \(a\) captures low-frequency, slow content, while small \(a\) resolves high-frequency, fast content;
\item \(b\) is the translation that slides the wavelet along the record;
\item \(\psi(t)\) is the mother wavelet—we use the analytic Morlet kernel \(\psi(t)=\pi^{-1/4}e^{j2\pi f_0 t}e^{-t^{2}/2}\).
\end{itemize}

With the Morlet kernel the centre frequency at scale \(a\) is \(f\approx f_0/a\).  Thus wavelet scales map directly to physical periods: for example, a \(30\;\mathrm{min}\) oscillation appears at \(f\approx(30\;\mathrm{min})^{-1}\).  Our tide-induced power variation falls between \(\sim10\;\mathrm{min}\) (fastest rise or fall) and
\(\sim120\;\mathrm{min}\) (near standstill), so we analyse only the corresponding tide band
\begin{equation}
\begin{aligned}
\mathcal{A} & = \Bigl\{\,a : f_0/a \in [\,f_{\min},\,f_{\max}] \Bigr\}, \\
f_{\min} & =\tfrac{1}{120}\,\text{Hz}, \qquad
f_{\max} =\tfrac{1}{10}\,\text{Hz}.
\end{aligned}
\label{eq:tideband}
\end{equation}
Energy outside \(\mathcal{A}\) arises from interference, boat wakes, or equipment noise and carries little information about the true water level.  Restricting the transform in \eqref{eq:cwt_def} to this band therefore focuses the analysis on exactly the scales that matter for tide estimation.

\subsection{Connection to the two–ray model}\label{sec:link_2ray}

Taking the magnitude square of the complex field in \eqref{eq:rx_cplx} and setting \(|\rho|\!\simeq\!1,\,
\phi_\rho\simeq\pi\) gives
\begin{equation}
P_{r,k}(h)
  = P_{0,k}\Bigl[\,1-\cos\!\bigl(2\pi \Delta d(h)/\lambda\bigr)\Bigr].
\label{eq:prh_repeat}
\end{equation}
With the far-field approximation \(\Delta d(h)\!\approx\!(2h/\!d)\,(h_t-h_r)\) the bracket in
\eqref{eq:prh_repeat} becomes a pure cosine in \(h\).  Absorbing the constant offset \(P_{0,k}\) into \(A\) and the cosine amplitude into \(B\) gives the simplified baseband envelope
\begin{equation}
p\bigl(h(t)\bigr)=A+B\cos \!\bigl(k\,h(t)\bigr),\qquad
k=\frac{2\pi\,(h_t-h_r)}{\lambda d}.
\label{eq:p_of_h}
\end{equation}
This expression contains only measurable geometric quantities \((h_t,h_r,d)\) and the carrier wavelength \(\lambda\).

\paragraph*{Signal model}
The recorded trace is
\begin{equation}
x(t)=p\bigl(h(t)\bigr)+\eta(t),
\label{eq:x_of_t}
\end{equation}
where \(\eta(t)\) lumps together higher-frequency clutter (boat wakes, wind ripples, RF interference, and receiver noise).

To relate the wavelet coefficients to the tide rate, we approximate the smooth envelope \(p\bigl(h(t)\bigr)=A+B\cos(kh(t))\) in the neighbourhood of a reference instant \(b\).  First we expand the tide height itself:
\begin{equation}
h(t)=h(b)+h'(b)\,\Delta t+\tfrac12 h''(b)\,\Delta t^{2}+O(\Delta t^{3}),
\qquad \Delta t\triangleq t-b .
\label{eq:h_taylor}
\end{equation}
Inserting \eqref{eq:h_taylor} into \(p(h)\) and retaining terms only up to first order in \(\Delta t\) gives
\begin{align}
p\!\bigl(h(t)\bigr)
 &= p\!\bigl(h(b)\bigr)
    +\frac{\mathrm{d}p}{\mathrm{d}h}\Big|_{h(b)}\,h'(b)\,\Delta t
    +O(\Delta t^{2})\notag\\
 &= p_0(b)+p_1(b)\,\Delta t+O(\Delta t^{2}),
\label{eq:p_taylor}
\end{align}
where
\begin{align}
p_0(b) &= A+B\cos\!\bigl(kh(b)\bigr),            \label{eq:p0_def}\\[4pt]
p_1(b) &= -Bk\,h'(b)\sin\!\bigl(kh(b)\bigr).      \label{eq:p1_def}
\end{align}

Here \(p_0(b)\) represents the instantaneous power level, while \(p_1(b)\) multiplies \(\Delta t\) and contains the tide velocity \(h'(b)\); the factor \(-Bk\sin(kh(b))\) modulates its sign and magnitude but remains finite throughout the tide cycle.  Quadratic and higher terms are \(O(\Delta t^{2})\).  For the largest wavelet scale we use (\(a\le 10\; \mathrm{min}\)) the Morlet window covers at most \(|\Delta t|\le 20\; \mathrm{min}\), so these higher-order terms are two orders of magnitude smaller than the linear term when compared with the tide period and can safely be neglected. 

\begin{figure*}[t]
  \centering
  \subfloat[Sensor record\label{fig:tide_vs_cwt:a}]{
    \includegraphics[width=0.45\linewidth]{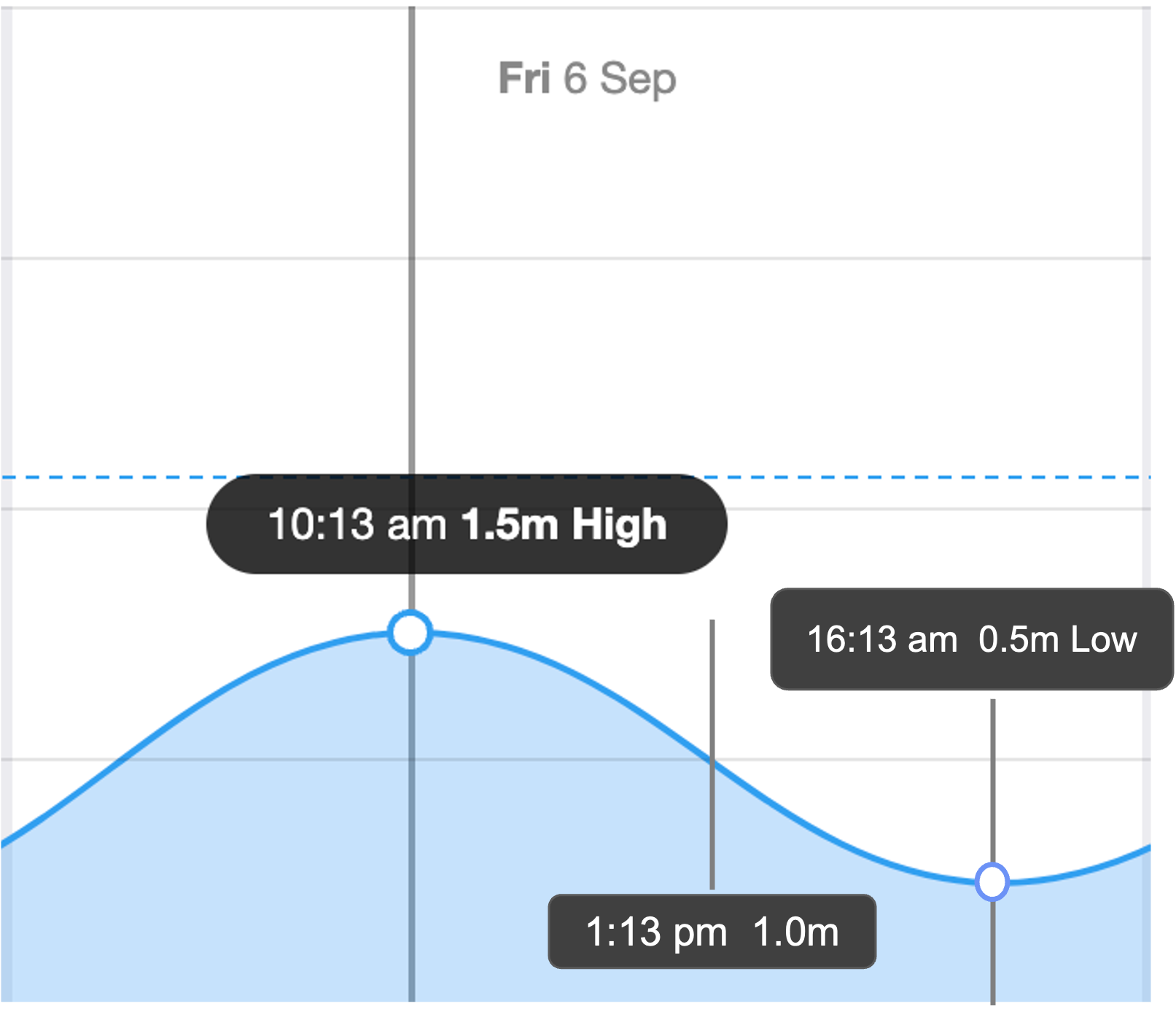}}
  \hfill
  \subfloat[CWT of RSRP of LTE signals\label{fig:tide_vs_cwt:b}]{
    \includegraphics[width=0.48\linewidth]{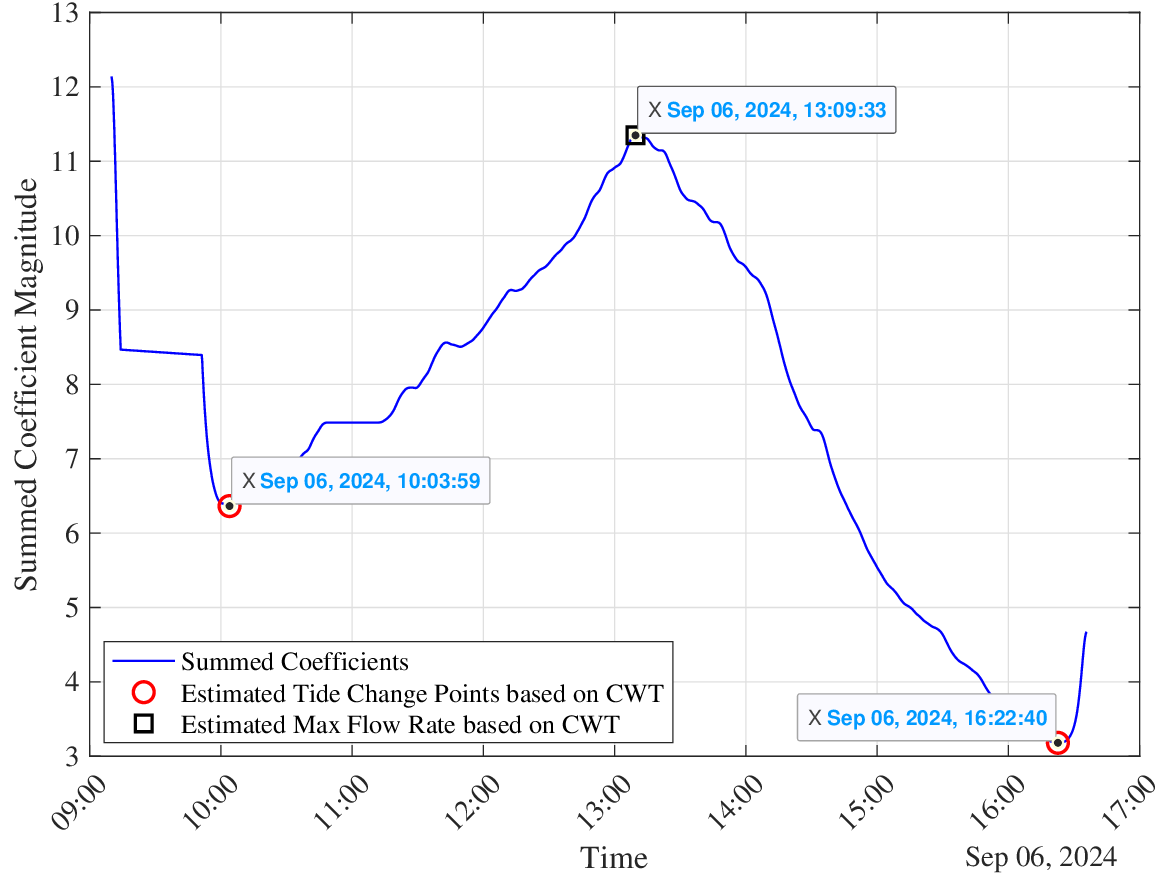}}
  \caption{(a) Sensor-measured tide height.  
           (b) Summed CWT magnitude $S(b)$ (blue); high/low detections
           (red circles) and maximum-flow detection (black square).}
  \label{fig:tide_vs_cwt}
\end{figure*}

\subsection{Link between wavelet magnitude and tide rate}\label{sec:lemma}

We now prove that the CWT magnitude is proportional to the local tide velocity \(h'(b)\).

\begin{lemma}\label{lem:rate}
For any scale \(a\in\mathcal{A}\) whose pseudo-period is small relative to the tidal semidiurnal cycle,
\begin{equation}
|W(a,b)| \;=\; C(a,b)\,|h'(b)| \;+\; O(a^{2}),
\label{eq:Wa_final}
\end{equation}
with
\begin{equation}
C(a,b)=B\,k\,|\sin(kh(b))|\,|M_1|\,a^{1/2},
\quad
M_1=\!\int_{-\infty}^{\infty}\!\tau\,\psi^{*}(\tau)\,d\tau .
\label{eq:Cab_final}
\end{equation}
Here \(k=2\pi(h_t-h_r)/(\lambda d)\) and \(B\) is the cosine amplitude in \eqref{eq:p_of_h}.
\end{lemma}

\begin{proof}
The CWT centred at \(b\) and scale \(a\) is
\begin{equation}
W(a,b)=\!\int_{-\infty}^{\infty}
        x(t)\,\psi_{a}^{*}(t-b)\,dt,
\; \psi_{a}(t)=a^{-1/2}\psi(t/a).
\label{eq:CWT_def}
\end{equation}
With \(x(t)=p(h(t))+\eta(t)\) and \(\eta(t)\) zero-mean,
\[
W(a,b)=\!\int p(h(t))\psi_{a}^{*}(t-b)\,dt + O(a^{2}).
\]
Insert the first-order envelope \(p(h(t))=p_0(b)+p_1(b)\Delta t+O(\Delta t^{2})\) from \eqref{eq:p_taylor}:
\[
\arraycolsep=2pt
\begin{array}{l}
W(a,b)=p_0(b)\!\int\!\psi_{a}^{*}(\Delta t)\,d\Delta t\\
\phantom{W(a,b)=}
          +p_1(b)\!\int\!\Delta t\,\psi_{a}^{*}(\Delta t)\,d\Delta t
          +O(a^{2}).
\end{array}
\]
The first integral is zero (Morlet has zero mean).  For the second, set \(\tau=\Delta t/a\):
\[
\int\!\Delta t\,\psi_{a}^{*}(\Delta t)\,d\Delta t
     =a^{3/2}\!\int_{-\infty}^{\infty}\tau\psi^{*}(\tau)\,d\tau
     =a^{3/2}M_1 .
\]
Hence \(W(a,b)=p_1(b)\,a^{3/2}M_1+O(a^{2})\).  Using \(p_1(b)=-Bk\,h'(b)\sin(kh(b))\) from \eqref{eq:p1_def} and taking
magnitudes gives \eqref{eq:Wa_final} with \(C(a,b)\) as in \eqref{eq:Cab_final}.
\end{proof}

Constant \(B\) and factor \(k\) are set by geometry; the wavelet moment \(|M_1|\) is fixed and non-zero; and \(|\sin(kh(b))|\) varies smoothly between 0 and 1.  Thus \(C(a,b)\) is bounded and slowly changing, so \eqref{eq:Wa_final} shows the CWT magnitude is driven mainly by the tide speed: it is smallest when tide is changing (\(h'(b)\approx0\)) and grows as the flow accelerates.

\subsection{Summed-coefficient feature}\label{sec:sumcoef}

Lemma \ref{lem:rate} states that for every scale \(a\in\mathcal{A}\) the magnitude \(|W(a,b)|\) is proportional to the absolute tide rate \(|h'(b)|\) up to a slowly varying factor.  Summing the magnitudes over all tide-band scales therefore reinforces the common \(|h'(b)|\) component while averaging out scale-specific clutter:
\begin{equation}
S(b)=\sum_{a\in\mathcal{A}} |W(a,b)|.
\label{eq:sumcoef}
\end{equation}
Because the scale factor varies only gently with \(a\), the curve \(S(b)\) inherits three key properties observed in every data set. First, local minima coincide with high and low water, where \(h'(b)\approx0\).  Second, broad maxima appear near mid-ebb and mid-flood, where the tide flows fastest.  Third, the feature is robust, i.e., slight timing gaps or narrowband interference at one scale are averaged out by the sum.

\subsection{Validation against a tide-sensor record}\label{sec:tide_match}

\paragraph{Online extraction.}  
By calculating $S(b_k)$ once per minute, Algorithm \ref{alg:online} buffers the last \SI{45}{min} of history, waits a fixed \SI{5}{min} to guard against future excursions, and tags the present sample as a tide turn if it is the lowest point in that window or as a peak flow if it is the highest. The procedure is single-pass and $O(N)$, so it runs easily in real time on an embedded SDR.

\begin{algorithm}[t]
\small
\caption{On-line detection of tide turn and peak flow}
\label{alg:online}
\begin{algorithmic}[1]
\Require Stream $\{S(b_k)\}$ sampled each minute;\\
         Look-back $T_{\text{back}} = \SI{45}{min}$;\\
         Look-ahead $T_{\text{fwd}} = \SI{5}{min}$
\For{$k = 1,2,\dots$}
    \State $\mathcal{W} \gets \{j : b_k - T_{\text{back}} \le b_j \le b_k + T_{\text{fwd}} \}$
    \If{$S(b_k) = \min\limits_{j \in \mathcal{W}} S(b_j)$}
        \State Label $b_k$ as {High/Low-Water}
        \State Skip the next \SI{5}{min} of samples
    \ElsIf{$S(b_k) = \max\limits_{j \in \mathcal{W}} S(b_j)$}
        \State Label $b_k$ as {Max-Flow}
    \EndIf
\EndFor
\end{algorithmic}
\end{algorithm}

Figure \ref{fig:tide_vs_cwt}(a) shows the water-level record from the Meadowbank Wharf tide sensor maintained by the New South Wales monitoring network\,\cite{willyweather}.  Figure \ref{fig:tide_vs_cwt}(b) plots the summed CWT magnitude $S(b)$ of \eqref{eq:sumcoef} derived from the \SI{2.6598}{GHz} LTE capture taken a few metres from that sensor on the same day.  Red circles mark the instants declared as \textsc{High/Low-Water} and the black square marks the \textsc{Max-Flow} instant produced by the online Algorithm \ref{alg:online}. The estimated times differ from the gauge times by less than the \SI{15}{min} reporting interval of the sensor, demonstrating that the radio-based method tracks the true tide with practical accuracy.

\subsection{Extension to multiple base stations}\label{sec:wavelet_multi}

The analysis above treated a single serving cell. In practice, several LTE cells (distinct carriers and/or sectors) may be observable at the same receiver and time. Each cell views the same tide trajectory \(h(t)\) but with its own slowly varying scale and disturbance pattern, set by geometry, carrier frequency, and local clutter. We therefore apply the CWT per cell and then fuse the resulting tide–band features on a common time grid.

Let \(\ell\in\{1,\dots,L\}\) index the detectable cells. For each cell \(\ell\), we construct an RSRP time series, compute the CWT \(W_{\ell}(a,b)\) over the tide band \(\mathcal{A}\) in \eqref{eq:tideband}, and form the per–cell magnitude sum
\begin{equation}
S_{\ell}(b)\;=\;\sum_{a\in\mathcal{A}} \bigl| W_{\ell}(a,b) \bigr| .
\end{equation}
To remove cell-specific offsets and scale, we standardise \(S_{\ell}\) over a rolling window \(\mathcal{W}_b\) (6–12\,h suffices in practice) using the median and median absolute deviation (MAD),
\begin{equation}
\tilde S_{\ell}(b)\;=\;
\frac{\,S_{\ell}(b)-\operatorname{median}_{u\in\mathcal{W}_b}\!S_{\ell}(u)\,}
     {\,\operatorname{MAD}_{u\in\mathcal{W}_b}\!S_{\ell}(u)\,}.
\end{equation}
All cells are placed on the same absolute time axis. If small micro-lags arise between cells (e.g., different look angles along the river), an optional integer shift \(\Delta_{\ell}\) can be estimated by short-window cross-correlation and applied to \(\tilde S_{\ell}(b)\); in our captures these shifts were negligible, so \(\Delta_{\ell}=0\) was used.

A fused tide-band feature is obtained by a robust aggregate over the available cells,
\begin{equation}
S_{\text{fused}}(b)\;=\;\operatorname{median}_{\ell}\,\tilde S_{\ell}(b),
\label{eq:S_fused}
\end{equation}
which preserves the characteristic morphology—minima near slack water and broad maxima near peak flow—while suppressing cell-specific disturbances (boats, foliage, transient interference). The online detector of Algorithm~\ref{alg:online} is then applied unchanged to \(S_{\text{fused}}(b)\) with the same \SI{45}{min} look-back and fixed \SI{5}{min} look-ahead, so moving from one to many cells does not alter latency; it increases robustness and spatial diversity.

\section{Experimental Study}\label{sec:exp}

\subsection{Measurement set-up}\label{ssec:exp_setup}

Figure \ref{fig:setup_montage} brings together the link geometry, surroundings, and hardware. The LTE macro-cell providing coverage over the river operates at \(f_c=2.6598\;\text{GHz}\) with a \SI{20}{MHz} channel and sits \(\approx\!\SI{35}{m}\) above ground on a small \(\approx\!\SI{10}{m}\)-high building \((-33.82528^{\circ}\,\text{S},\,151.09172^{\circ}\,\text{E})\).
Across the Parramatta River, \SI{420}{m} away, a four-element log-periodic array receives the downlink (Fig.~\ref{fig:setup_montage}\,c).  The elements face the base station antennas, are mounted vertically, and are tilted \(27^{\circ}\) downward toward the water; their feed-point heights above the ground are \(h_{r,1{:}4}=\{1.95,\,1.80,\,1.65,\,1.50\}\,\text{m}\) (see the side-elevation sketch in Fig.~\ref{fig:setup_montage}\,b).
Table~\ref{tab:geo} lists the key parameters.

\begin{table}[t]
  \centering
  \caption{Key link parameters}
  \label{tab:geo}
  \begin{tabular}{@{}lcc@{}}
    \toprule
    Quantity & Symbol & Value \\ \midrule
    Carrier frequency & \(f_c\) & \SI{2659.8}{MHz} \\
    LTE bandwidth & \(B\) & \SI{20}{MHz} \\
    Tx height (centre) & \(h_t\) & \(\approx\)\SI{45}{m}\\
    Rx heights & \(h_{r,1{:}4}\) & \{1.95, 1.80, 1.65, 1.50\}\,m\\
    Horizontal range & \(d\) & \SI{420}{m}\\
    Snapshots per min & – & \(\approx\)10\\
    Total captures & – & 4216 (\(\approx\) 7 h)\\ \bottomrule
  \end{tabular}
\end{table}

Each antenna feeds a \SI{10}{dB} LNA and an FMCOMMS5/Zynq four-channel SDR; complex baseband I/Q is saved to \texttt{.bb} files at roughly ten snapshots per minute.   A short MATLAB routine reads each file, down-converts, performs LTE cell-search and channel decoding, and logs the per-antenna RSRP, RSSI, and RSRQ values.

Water level was measured on site every \SI{10}{min} by the first author and compared with the public tide record from the New South Wales monitoring network (Meadowbank Wharf station)~\cite{willyweather}.  

Before building the learning features we inspected the raw metrics to verify that the received signal clearly responds to changes in water level.  All records with zero RSRP, RSSI, or RSRQ were discarded and the powers converted from dBm/dB to linear scale.  Remaining samples were cleaned with an IQR test (\(k{=}1\)) followed by a Hampel filter.  Figure~\ref{fig:raw_metrics} shows the resulting RSRP, RSSI, and RSRQ of Antenna 1 plotted against the measured
tide height.  The near-sinusoidal envelopes—especially evident in RSRP and RSSI — demonstrate how strongly the signal strength is affected by the rising and falling tide.
\begin{figure*}[t]
  \centering
  \begin{subfigure}[b]{0.29\textwidth}
    \centering
    \includegraphics[width=\linewidth]{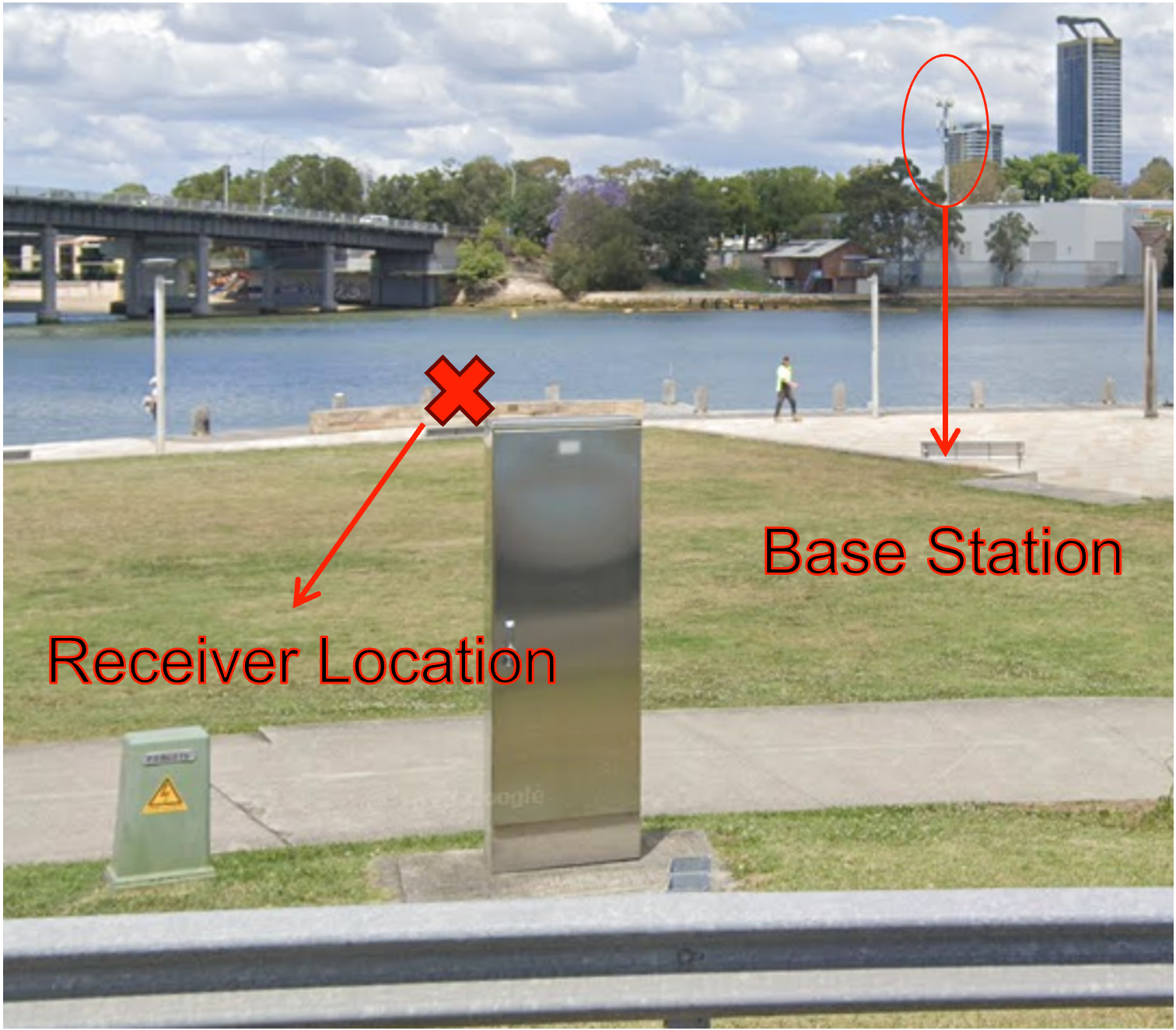}
    \caption{River-bank environment with base station tower (right) and receiver site (left).}
    \label{fig:field_view}
  \end{subfigure}
  \begin{subfigure}[b]{0.45\textwidth}
    \centering
    \includegraphics[width=\linewidth]{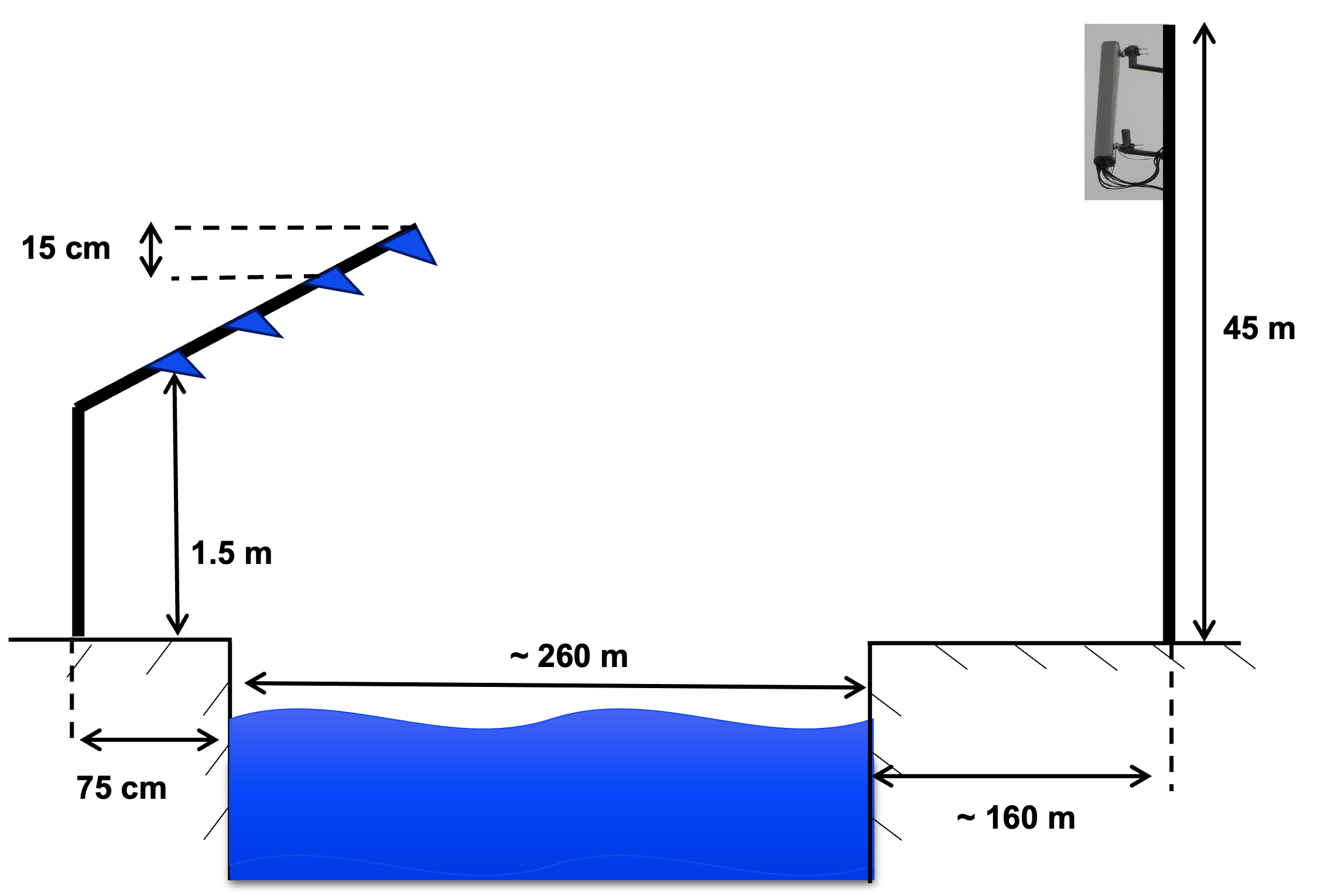}
    \caption{Tx–Rx elevation schematic.}
    \label{fig:geometry}
  \end{subfigure}
  \begin{subfigure}[b]{0.22\textwidth}
    \centering
    \includegraphics[width=\linewidth]{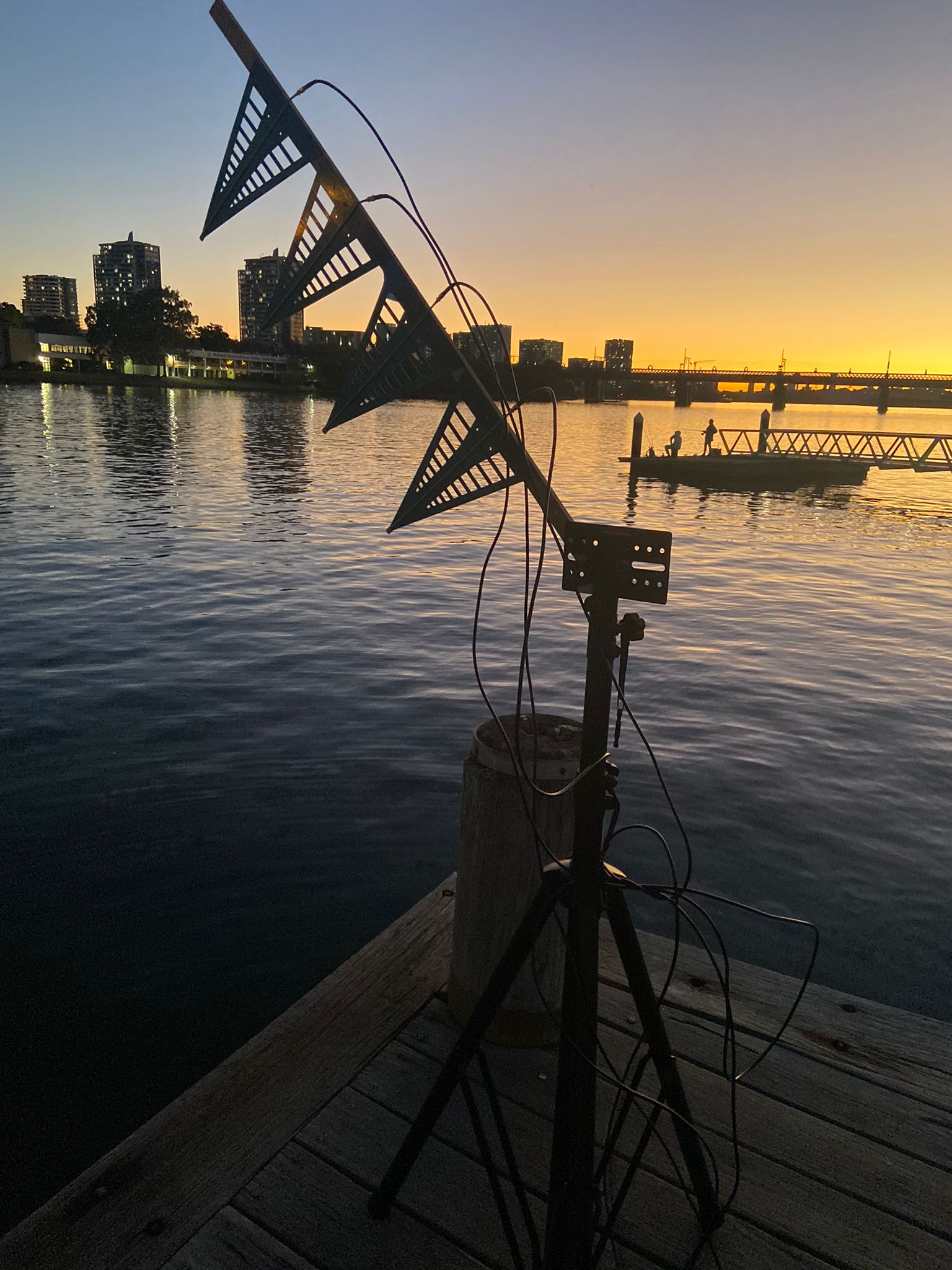}
    \caption{Four-element log-periodic receive antenna array.}
    \label{fig:antenna_photo}
  \end{subfigure}
  \caption{Measurement set-up: (a) environment, (b) geometry,
           and (c) hardware.}
  \label{fig:setup_montage}
\end{figure*}

\begin{figure*}[t]
  \centering
  \begin{subfigure}[b]{.32\textwidth}
    \centering\includegraphics[width=\linewidth]{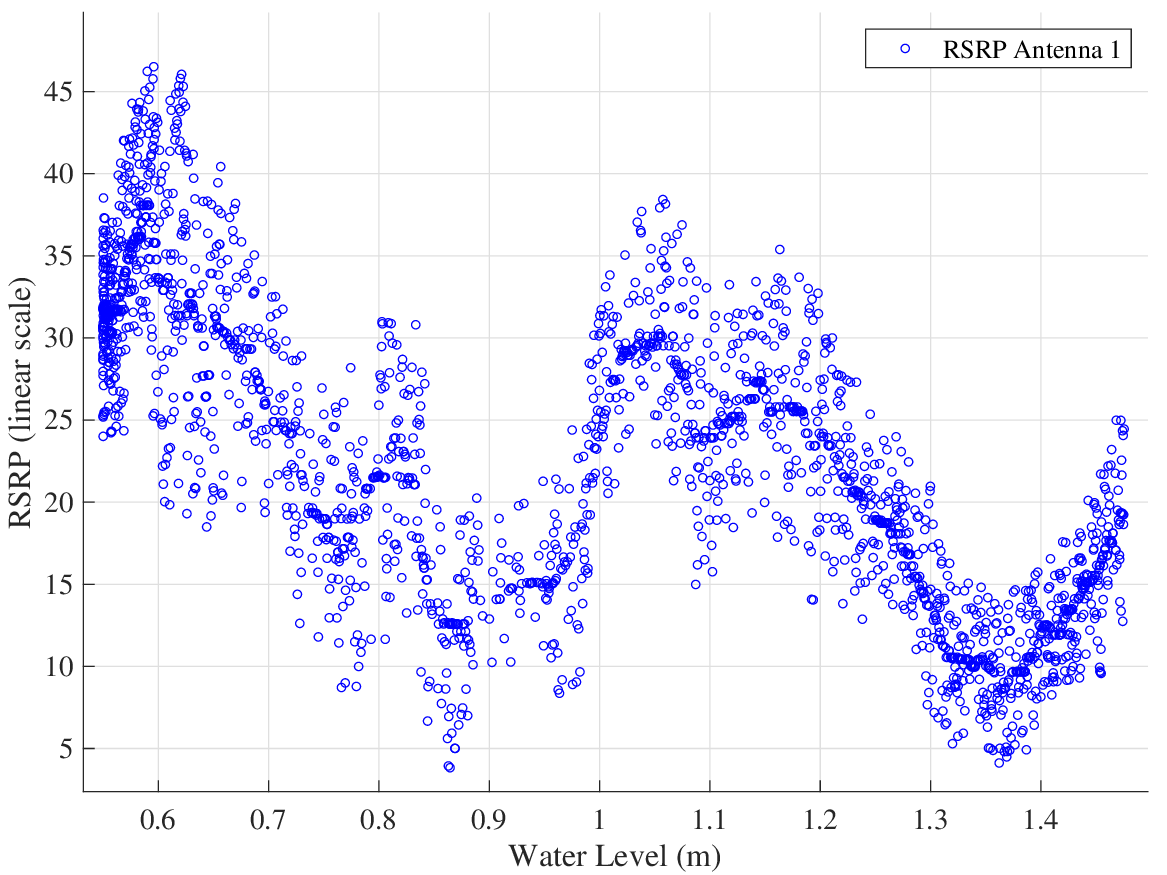}
    \caption{RSRP versus water level.}
    \label{fig:raw_rsrp}
  \end{subfigure}
  \begin{subfigure}[b]{.32\textwidth}
    \centering\includegraphics[width=\linewidth]{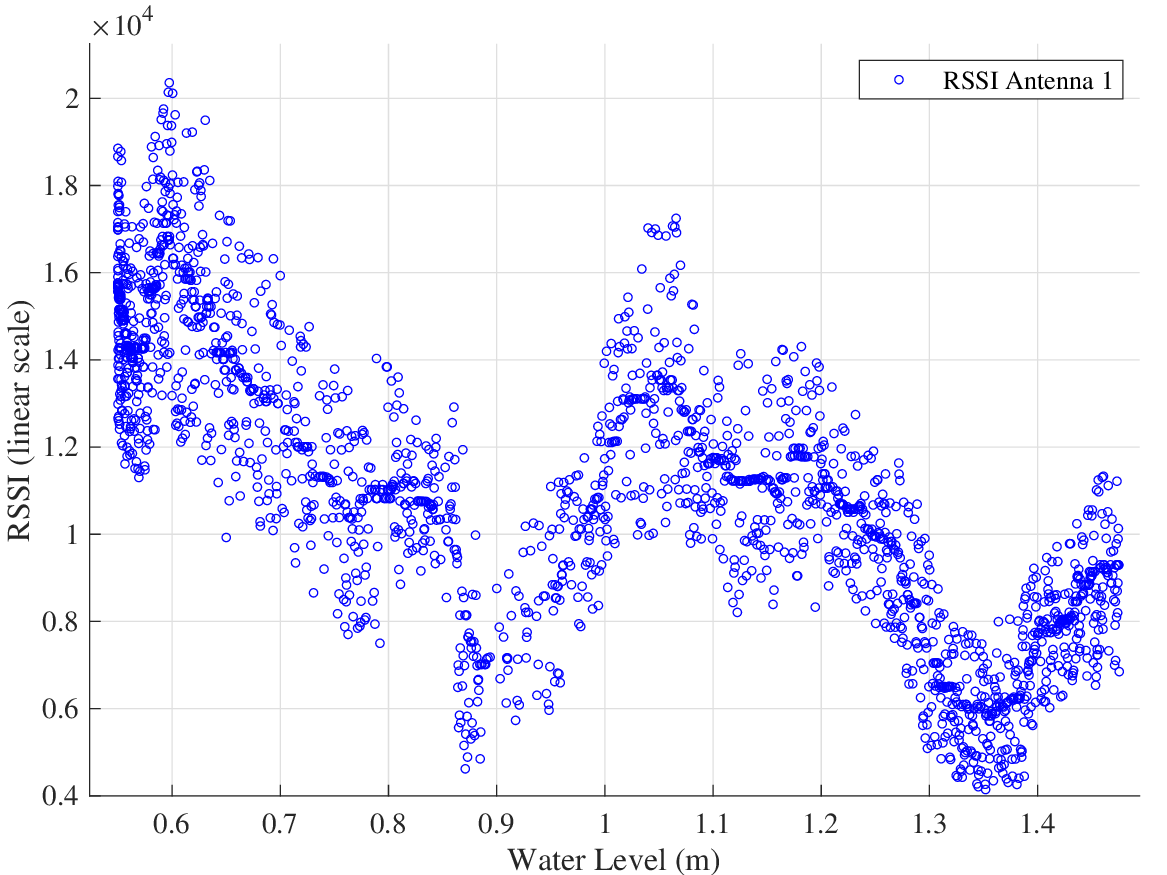}
    \caption{RSSI versus water level.}
    \label{fig:raw_rssi}
  \end{subfigure}
  \begin{subfigure}[b]{.32\textwidth}
    \centering\includegraphics[width=\linewidth]{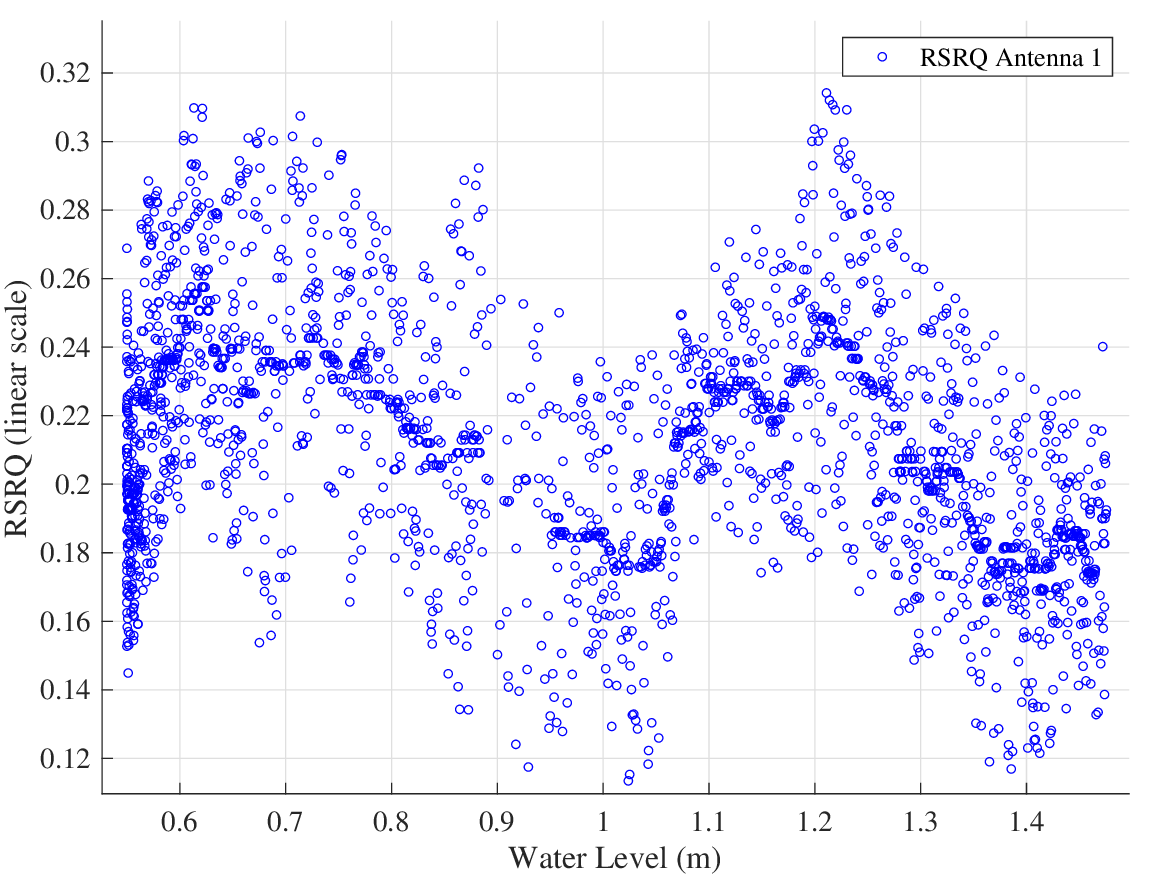}
    \caption{RSRQ versus water level.}
    \label{fig:raw_rsrq}
  \end{subfigure}
  \caption{Scatter of cleaned signal metrics for Antenna 1 versus water level on 6 Sep 2024. }
  \label{fig:raw_metrics}
\end{figure*}

\begin{figure}[t]
  \centering
  \includegraphics[width=\linewidth]{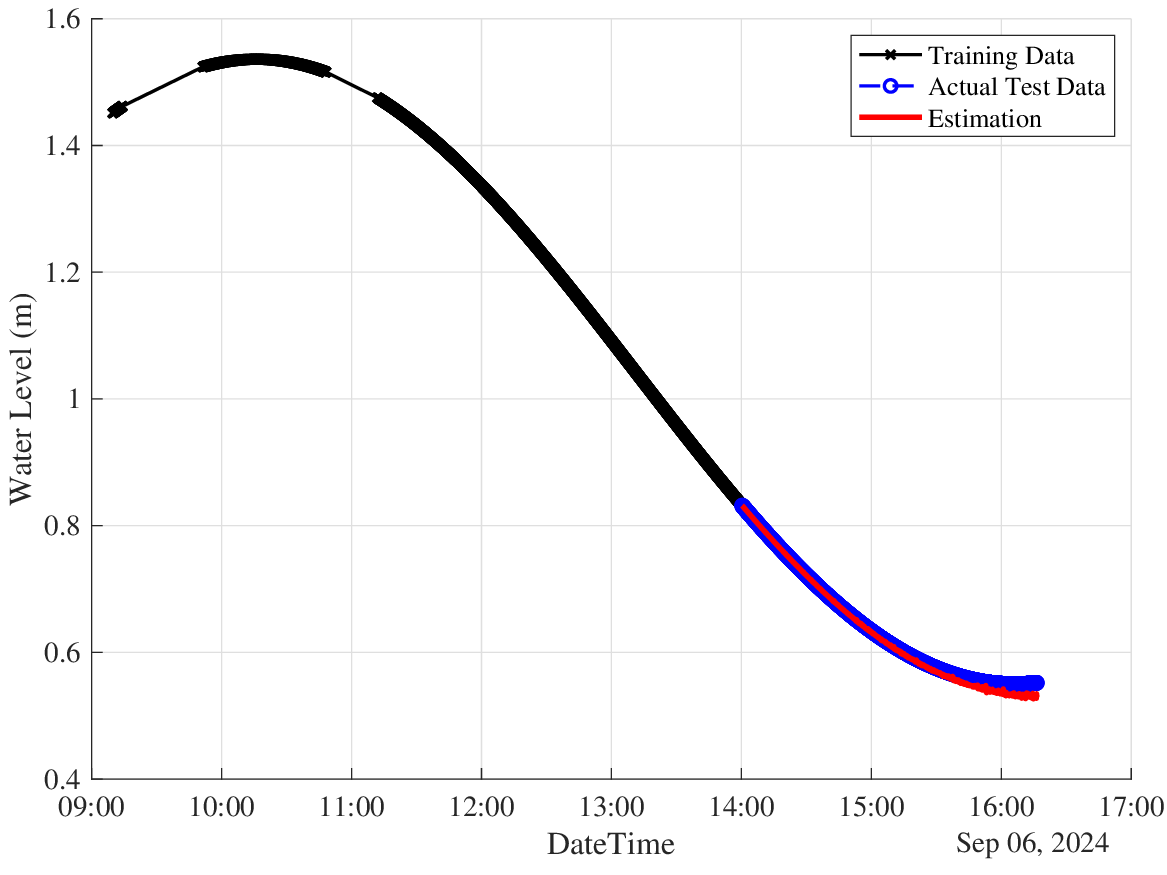}
  \caption{Water-level estimation on 6 Sep 2024. Black \(\times\): training segment; blue dashed: unseen test segment;
           red: water level estimation. Grey gaps correspond to missing SDR files.}
  \label{fig:nn_pred_sep6}
\end{figure}

\subsection{Learning model and feature set}\label{ssec:nn}

Each baseband file is decoded to obtain the four-antenna RSRP, RSSI and RSRQ traces. From these we build a feature vector that contains the per-antenna power readings, every pairwise difference and ratio (more robust to slow
transmitter drift than the raw levels), and the sine and cosine of the semidiurnal tide phase \(\varphi(t)=2\pi t/T_{\text{tide}}\) with \(T_{\text{tide}}\approx\SI{12.6}{h}\).  The phase is extracted by applying the CWT to the RSRP of the signal and computing its scalogram. By identifying and selecting the frequency component that remains dominant over time—corresponding to the tidal cycle—we can trace the phase evolution. Since the tidal period varies across different locations, this approach enables automatic, location-specific extraction of the tide cycle directly from the captured signal data. All features are standardised to zero mean and unit variance before learning.
The regressor is a single-hidden-layer feed-forward neural network with 40 hyperbolic-tangent neurones and a linear output, trained once on the first 60 \% of the record and then applied in a streaming fashion to the remainder.

\subsection{Single-day evaluation}\label{ssec:within}

The seven-hour record of 6 September 2024 comprises 4216 snapshots, roughly ten captures per minute. Snapshots are ordered in time; the first 60 \% form the training set, the next 5 \% provide an early-stopping check, and the final 35 \% are held out for test.  After training, each new snapshot is converted to features, passed through the network, and an online water-level estimate is produced. The overall latency is determined by the \SI{5}{min} forward window required by the wavelet procedure.

Figure \ref{fig:nn_pred_sep6} compares the neural estimate with the reference water level reported by the New South Wales monitoring network \cite{willyweather}.  Over the unseen segment the model attains a root-mean-square error of {0.8\;cm} and a mean-absolute error of {0.5\;cm}.  These sub-centimetre errors, achieved without site-specific tuning, demonstrate that the wavelet-guided feature set, coupled with a modest neural network, is sufficient for precise, real-time water-level tracking.

\begin{figure*}[t]
  \centering
  \begin{subfigure}[b]{0.37\linewidth}
    \centering
    \includegraphics[width=\linewidth, trim=0 0 0 350pt, clip]{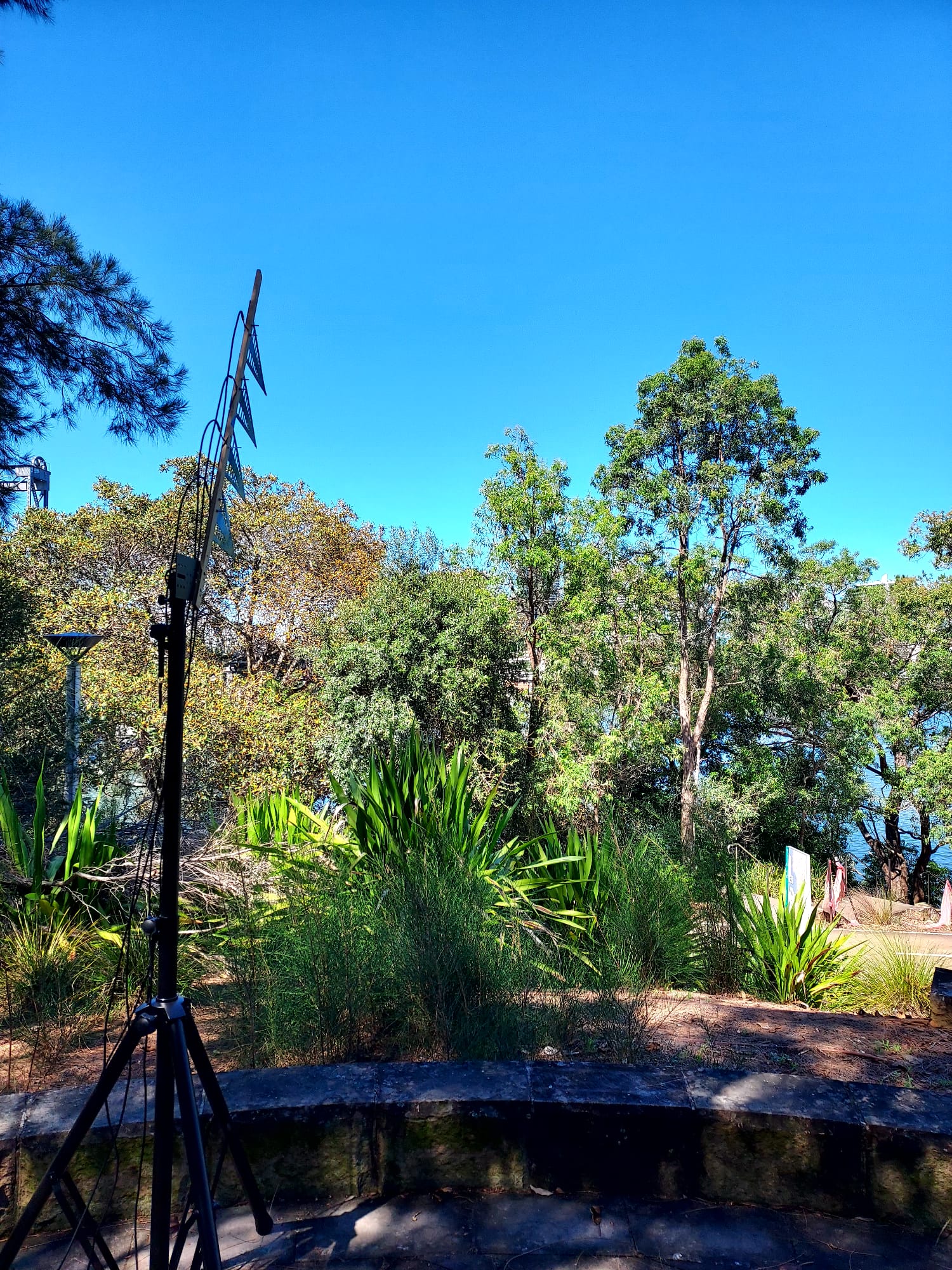}
    \caption{Receiver site on 2 Dec 2024 — trees obstruct the line-of-sight.}
    \label{fig:cross_env}
  \end{subfigure}
  \hfill
  \begin{subfigure}[b]{0.58\linewidth}
    \centering
    \includegraphics[width=\linewidth]{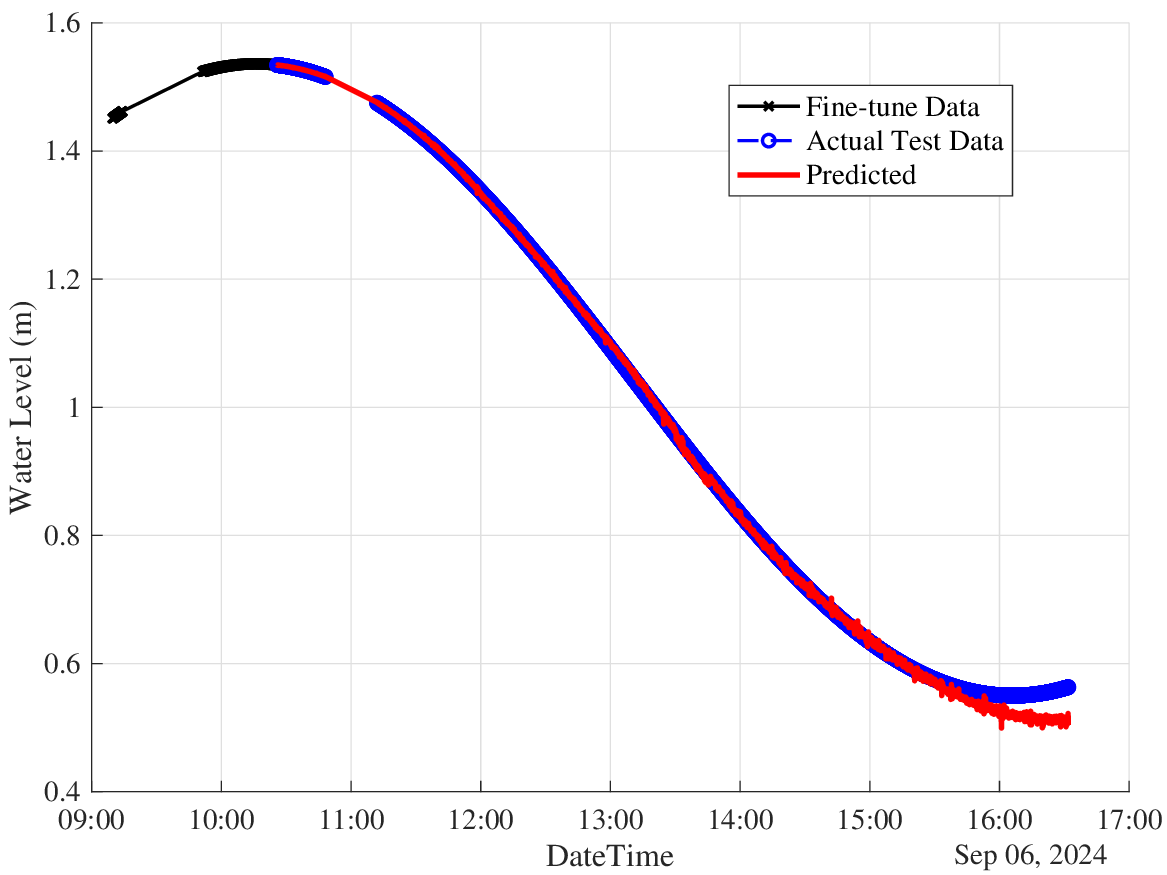}
    \caption{Model trained on 2 Dec, tested on 6 Sep.}
    \label{fig:cross_pred}
  \end{subfigure}
  \caption{Generalisation across measurement dates and propagation conditions.}
  \label{fig:cross_results}
\end{figure*}
\subsection{Cross-date validation: training 2 Dec,
            testing 6 Sep}\label{ssec:crossdate}

To verify that the model remains effective without a clear line-of-sight reflection, a second data set was acquired on 2~December~2024 at a point approximately \SI{510}{m} from the same base station and \SI{41}{m} back from the river bank.  The line-of-sight between the array and the water was obstructed by trees, and ferries passed the site frequently during the six-hour recording (see Fig.,\ref{fig:cross_results}a). The four log-periodic elements were mounted at heights of 2.45, 2.30, 2.15, and \SI{2.00}{m} and tilted \(27^{\circ}\) toward the river.  A total of 2\,933 baseband captures were logged on the usual \(f_c=\SI{2.6598}{GHz}\), \SI{20}{MHz} LTE carrier.

The neural network architecture and feature pipeline are kept unchanged. The network is first trained on the full 2 Dec record ("train-from-scratch"). To adapt to the new date without seeing the entire 6 Sep data stream, the first 10 \% of 6 Sep snapshots (about 410 files) are used for a brief fine-tuning step; the remaining part of data are then fed to the network in chronological order and the output is compared with the tide sensor reference.

Figure \ref{fig:cross_pred} shows the result. Despite dense vegetation and repeated ferry wakes, the network follows the true tide with a root-mean-square error of {1.7\;cm} and a mean-absolute error of {0.8\;cm} over more than five hours of previously unseen data.  The performance remains within a centimetre of the line-of-sight benchmark of Section~\ref{ssec:within}, demonstrating that the wavelet-guided feature set captures geometry-invariant structure and allows the small neural regressor to transfer successfully between dates, sites, and propagation conditions.

\subsection{Long-duration passive capture and CWT validation at 763\,MHz}
\label{sec:long-duration-cwt}

At our initial two riverbank deployments (6~Sep. and 2~Dec.), battery-powered operation limited continuous captures to 6–8\,h. Using the CWT, these datasets already reveal the semidiurnal tidal component ($\approx$12.6\,h). Nevertheless, because they cover less than a full tidal cycle, the evidence is partial and susceptible to boundary effects. To quantify the period conclusively from data alone and to support model design without external priors, we therefore conducted an additional long-duration passive capture at a separate site with stable mains power and co-located sonar water-level ground truth, complementing the earlier deployments.

We deployed a passive receiver at the UTS Haberfield Rowing Club on the Parramatta River (Fig.~\ref{fig:haberfield-setup}). The receiver is an SDR connected to two indoor monopole antennas mounted on a whiteboard, operating in NLoS relative to surrounding infrastructure. A high-precision sonar water-level sensor was installed under the wharf to provide ground truth. Although multiple LTE cells across several bands were observable at this location, all analysis here uses Band~28 (763\,MHz), which provided the strongest and most stable measurements. This operator-independent, indoor passive deployment—requiring no interaction with the cellular network or site-specific infrastructure—enabled multi-hour to multi-day captures spanning full tidal cycles.

\begin{figure}[t]
  \centering
  \includegraphics[width=\linewidth]{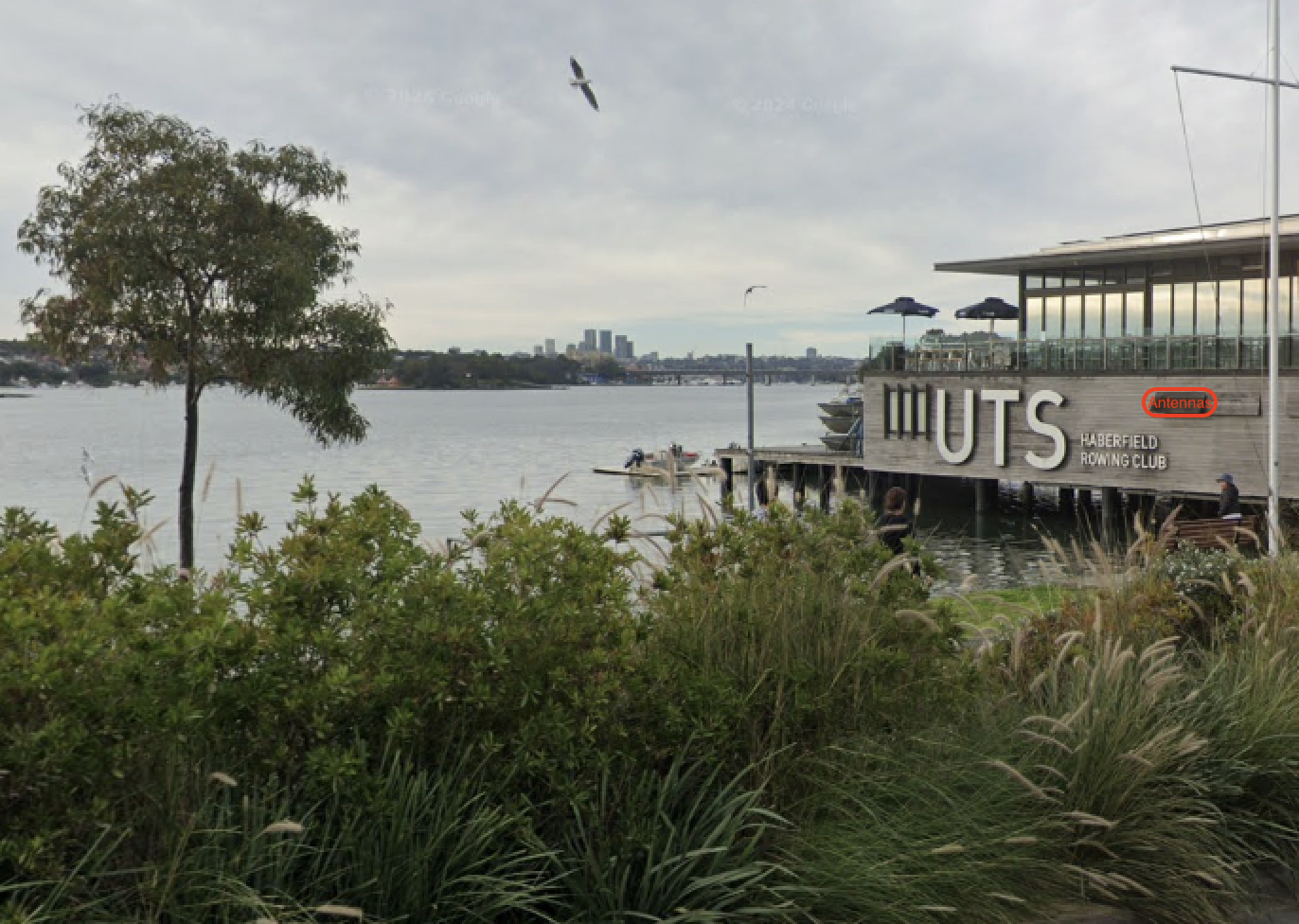}
  \caption{Long-duration passive capture site at the UTS Haberfield Rowing Club. An indoor SDR with two monopole antennas records LTE downlink signals (Band~28, 763\,MHz), while a sonar gauge under the wharf provides water-level ground truth.}
  \label{fig:haberfield-setup}
\end{figure}

Both the sonar water level and LTE time series were uniformly resampled to a 210\,s grid and preprocessed with robust outlier suppression and low-pass filtering. We then computed the CWT of the tide signal. The scalogram exhibits a clear energy ridge in the mili-Hertz (mHz) range associated with semidiurnal tides (Fig.~\ref{fig:cwt-tide-763}). The peak frequency $f_{\text{peak}}$ lies near $0.022$\,mHz, implying
\[
T \,=\, \frac{1}{f_{\text{peak}}} \,\approx\, 12.6~\text{h},
\]
which matches the expected semidiurnal tide at this site. This confirms that long-duration captures here allow direct measurement of the tidal periodicity from data without external priors.

\begin{figure}[t]
  \centering
  \includegraphics[width=\linewidth]{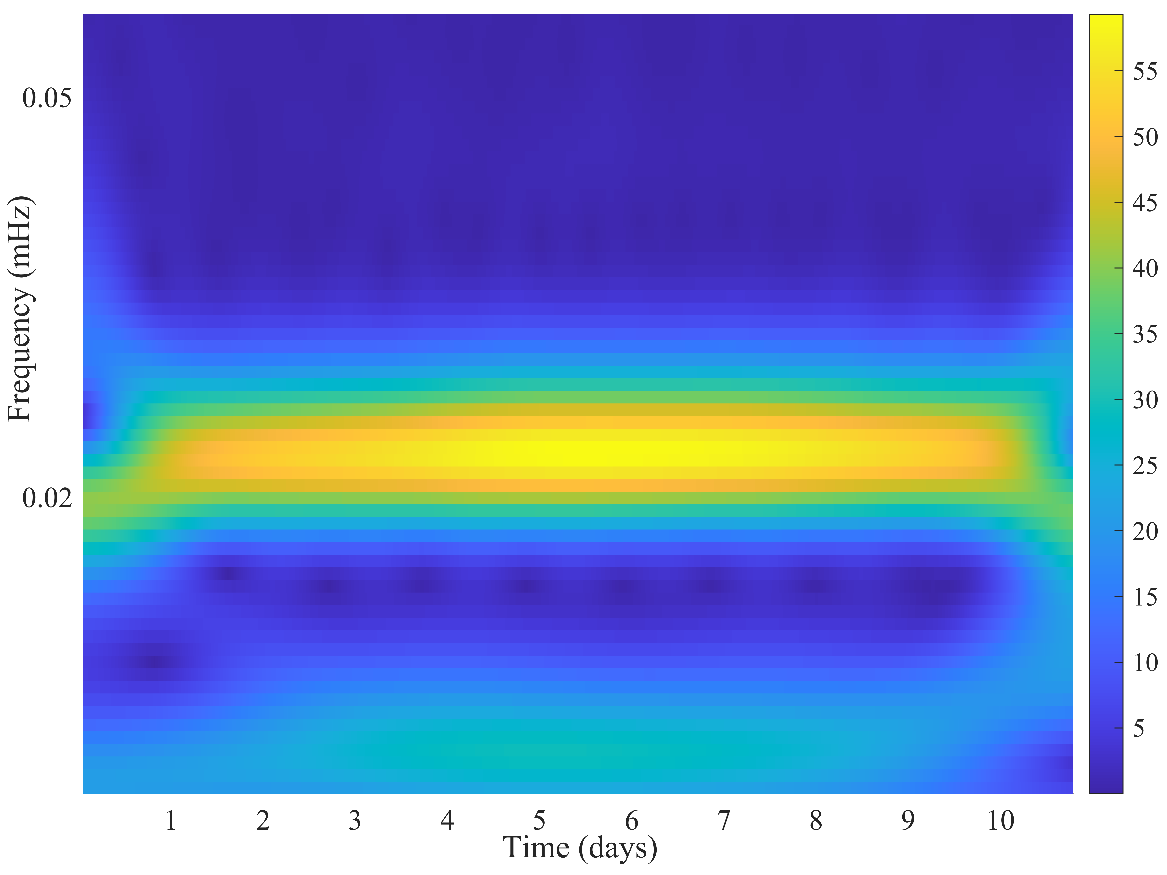}
  \caption{CWT of the tide series during the long-duration capture (763\,MHz session timeline). A strong ridge around $0.022$\,mHz corresponds to a $\sim$12.6\,h semidiurnal tide.}
  \label{fig:cwt-tide-763}
\end{figure}

Applying the identical preprocessing and CWT to the RSRP series at 763\,MHz reveals a distinct low-frequency ridge in the mHz range (Fig.~\ref{fig:cwt-rsrp-763}). In a CWT scalogram, a ridge denotes a contiguous time–frequency track of elevated coefficient magnitude, evidencing a persistent oscillatory component. Here, the ridge indicates that cyclic changes in water level modulate the path difference between direct and water-reflected propagation, producing a corresponding cyclic variation in received power. We do not assume equality between the radio and tide periods; it suffices that a stable low–mHz component is present in RSRP as expected from tide-driven multipath. These observations motivate the use of wavelet-derived features and cyclic structure in the learning pipeline.

\begin{figure}[t]
  \centering
  \includegraphics[width=\linewidth]{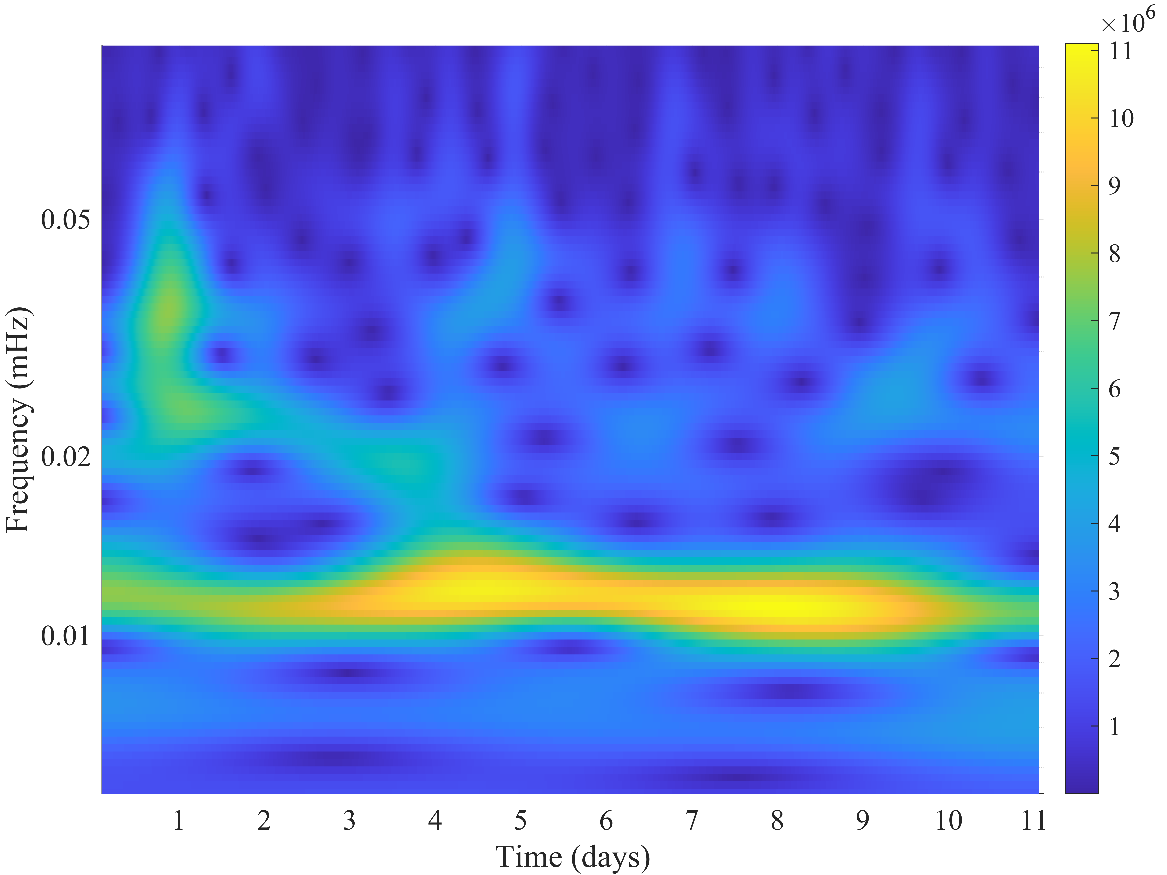}
  \caption{CWT of RSRP at 763\,MHz during the long-duration capture. A persistent low-frequency ridge (contiguous band of high wavelet magnitude) appears in the millihertz range, indicating a stable oscillatory component in received power driven by tide-induced multipath. }
  \label{fig:cwt-rsrp-763}
\end{figure}

\begin{figure}[t]
  \centering
  \includegraphics[width=\linewidth]{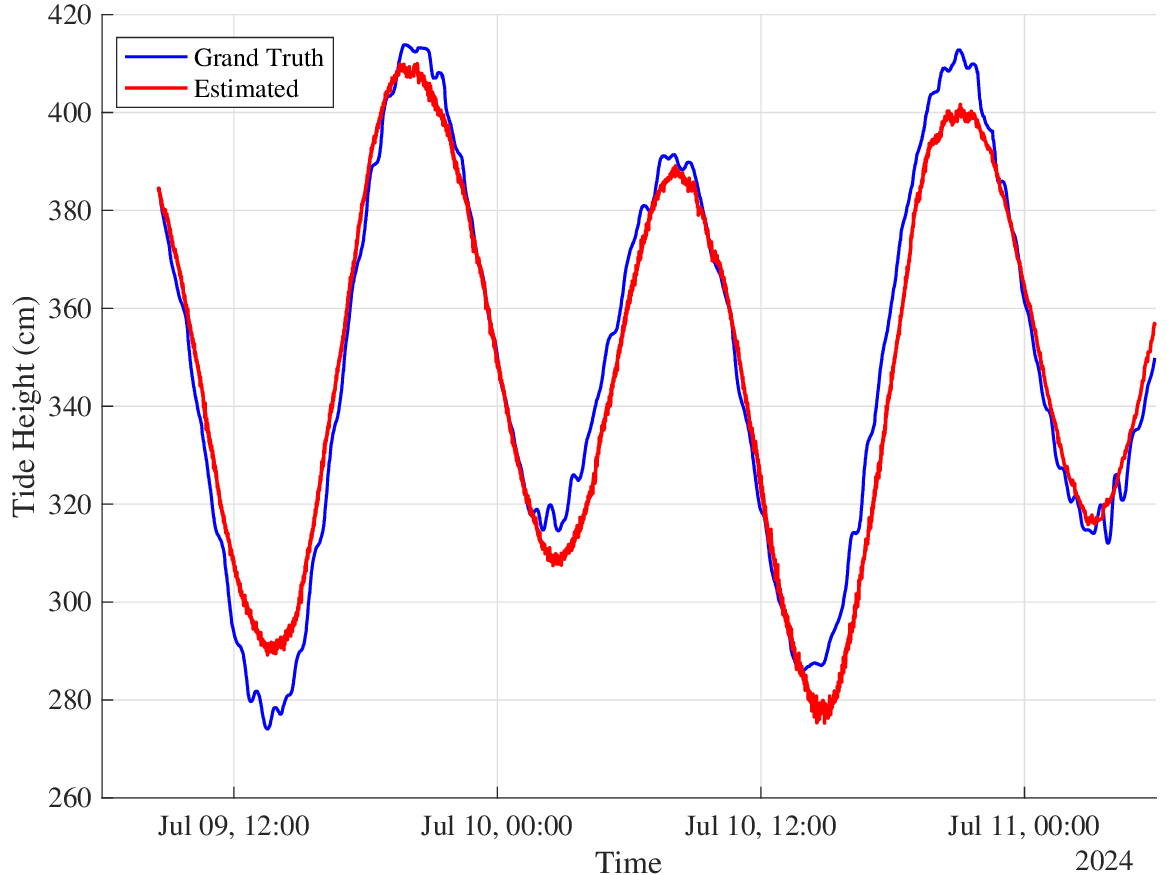}
  \caption{Neural-network water-level estimation on the 763\,MHz long-duration dataset. Estimated water level (red) versus sonar ground truth (blue) over the two-day test period. The model reproduces the tidal variation; higher-frequency perturbations from waves and boat wakes are visible in the ground truth.}
  \label{fig:NN763}
\end{figure}

Using the same 763\,MHz long-duration dataset, we trained a compact feed-forward neural network to estimate water level from LTE link metrics (two antennas) and time features. The data were split chronologically into 75\% training, 10\% validation, and 15\% test. Model selection was based on the validation set; the final model was then trained on train+validation and applied once to the held-out two-day test period. On this test set the model achieved RMSE$=8.785$\,cm and MAE$=7.257$\,cm (train+validation RMSE$=1.262$\,cm, MAE$=1.012$\,cm). The estimates closely follow the ground truth and capture the tidal envelope; short-timescale variations visible in the sonar record, including waves generated by boat traffic, appear as higher-frequency perturbations in the ground truth. Figure~\ref{fig:NN763} illustrates the estimated water level (red) against the ground truth (blue) over the test interval.

\subsection{Multi–base-station wavelet fusion for tide–flow and turn detection}
\label{ssec:multi-cell-exp}

To assess robustness when several base stations are visible at the same receiver, we used the site of Section~\ref{sec:long-duration-cwt} and monitored two downlink carriers from different operators: Optus at 1857.5\,MHz and Vodafone at 2117.5\,MHz (Fig.~\ref{fig:multi_env}). For each base station we formed the RSRP time series, computed the CWT over the tide band, and summed the wavelet magnitudes across scales to obtain a per–base-station tide feature $S_\ell(b)$. Each $S_\ell(b)$ was standardised with a robust 9\,h rolling median/MAD and aligned on a common time axis. The fused curve $S_{\text{fused}}(b)$ was then taken as the pointwise median across the two base stations and resampled to a 1\,min grid. 

Figure~\ref{fig:multi_sfused} plots $S_{\text{fused}}(b)$ against the co-located tide record. The fused tide-band feature preserves the characteristic shape seen with a single base station—peaks when the flow is fastest and minima when the tide is high/low (tide-change times)—while suppressing base-station–specific disturbances (boats, clutter, or sector scheduling differences). The same on-line detector of Algorithm~\ref{alg:online} is applied directly to $S_{\text{fused}}(b)$.

\begin{figure*}[t]
  \centering
  \begin{subfigure}[b]{0.43\linewidth}
    \centering
    \includegraphics[width=\linewidth]{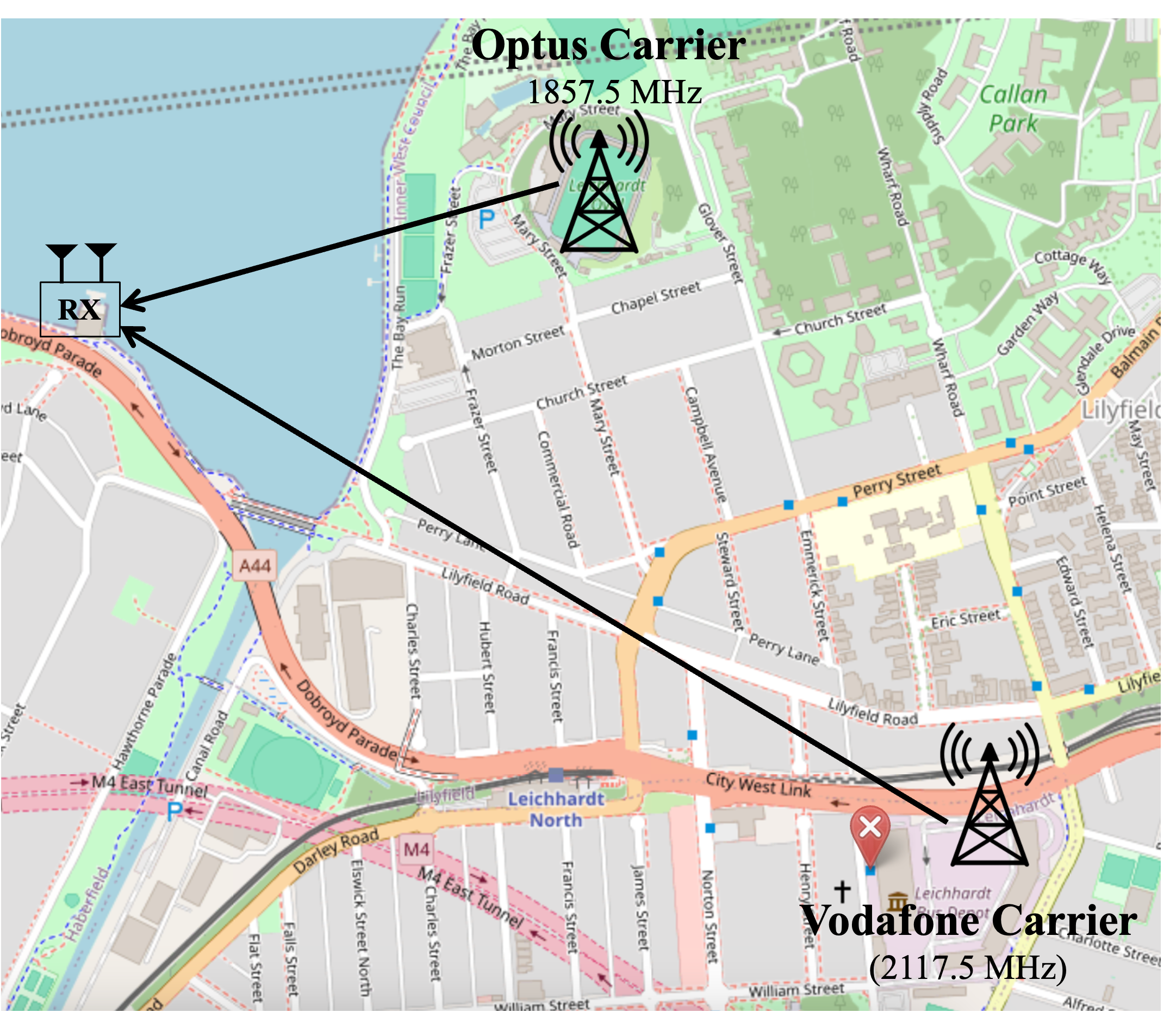}
    \caption{Receiver and two serving base stations (Optus 1857.5\,MHz; Vodafone 2117.5\,MHz).}
    \label{fig:multi_env}
  \end{subfigure}\hfill
  \begin{subfigure}[b]{0.53\linewidth}
    \centering
    \includegraphics[width=\linewidth]{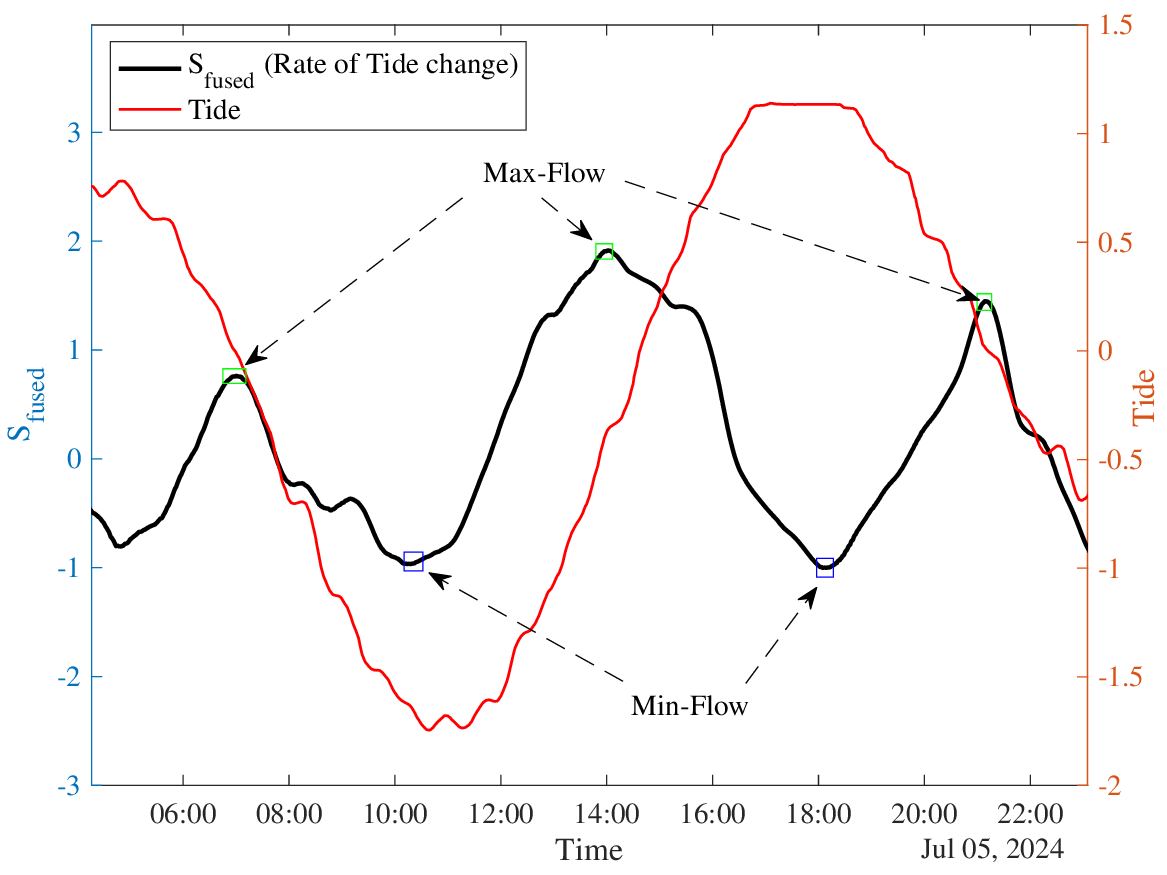}
    \caption{Fused tide-band feature $S_{\text{fused}}$ and tide record. $S_{\text{fused}}$ reflects the tide-rate: peaks mark maximum-flow and minima coincide with high/low tide (tide turns).}
    \label{fig:multi_sfused}
  \end{subfigure}
  \caption{Multi–base-station wavelet fusion. Independent CWTs per carrier are median-fused to obtain a single tide-band feature that tracks the tide-rate (flow speed) over time.}
  \label{fig:multi_results}
\end{figure*}

For learning, per–base-station features (RSRP/RSSI/RSRQ levels and within-base-station differences/ratios) can be z-scored per base station and concatenated together with $S_{\text{fused}}(b)$. If a base station is intermittently unavailable, its features are zero-imputed and flagged by a binary availability mask so the regressor benefits from diversity when present and remains stable when only a subset is visible. In this paper, however, all reported neural-network results (Section~\ref{ssec:nn}) use a single serving base station; they already achieve strong accuracy, so we did not need multi–base-station learning for the main results. We nevertheless show the wavelet-level multi–base-station fusion and its outcomes here to demonstrate how the approach extends when multiple carriers are available.

\section{Conclusion}\label{sec:concl}
This paper presented a physics-guided, wavelet-based framework that converts standards-compliant LTE power metrics (RSRP, RSSI, RSRQ) into centimetre-scale water-level estimates. A continuous wavelet transform isolates tide-driven structure in the radio signal and produces a summed-coefficient feature that marks high/low-tide times and captures the tide-rate (flow speed) with low latency. Using only receiver-side measurements and commodity hardware—without CSI, array calibration, or network cooperation—the method delivers continuous water-level tracking.

Beyond the two riverbank deployments, we conducted an operator-independent long-duration study at a different location (UTS Haberfield), using a different carrier frequency and an indoor passive setup with co-located sonar ground truth. The wavelet scalogram of the tide record recovers the semidiurnal period directly from data, and a compact model driven by LTE power features tracks the tide over a two-day held-out interval. These results indicate that the approach is not tied to a particular site, frequency, or capture configuration and can operate with commodity receivers in a fully passive manner.

We also demonstrated a multi–base-station extension at the wavelet level. Independent CWTs from two operators were median-fused to produce a single tide-band feature that continues to indicate maximum-flow peaks and high/low-tide minima while smoothing base-station-specific disturbances. This shows how the framework scales to cooperative sensing settings where multiple carriers are available.

We have shown that, across different locations, systems, frequencies, and setups, accurate water-level estimation is achievable from downlink LTE power metrics without CSI, and we demonstrated a wavelet-level multi-base-station fusion that robustly detects tide turns and peak-flow times. Building on this, our next steps are to (i) extend the same wavelet-guided approach to 5G NR using standard NR power metrics while keeping the CWT/summed-coefficient core unchanged; (ii) add water-level estimation from uplink signals as an independent channel to increase resilience; and (iii) integrate rainfall estimation into the model—using downlink signal behaviour—so the system jointly reports water level and rain.





\bibliographystyle{IEEEtran}

\end{document}